\documentclass[a4paper]{article}
\usepackage{amsmath}
\usepackage{amssymb}
\usepackage{enumerate}
\usepackage{xypic}
\usepackage{amsthm}
\usepackage{graphicx}
\usepackage{url}
\allowdisplaybreaks[4]

\newtheorem{theorem}{Theorem}[section]
\newtheorem{lemma}[theorem]{Lemma}
\newtheorem{proposition}[theorem]{Proposition}
\newtheorem{corollary}[theorem]{Corollary}

\newtheorem{definition}[theorem]{Definition}

\makeatletter
\providecommand{\bigsqcap}{%
  \mathop{%
    \mathpalette\@updown\bigsqcup
  }%
}
\newcommand*{\@updown}[2]{%
  \rotatebox[origin=c]{180}{$\m@th#1#2$}%
}
\makeatother

\newcommand{\jump}[1]{\ensuremath{[\![#1]\!]} }
\newcommand{\floor}[1]{\lfloor #1 \rfloor}

\newcommand{\PD}[4]{{\bf PT}_{#1,#2}(#3,#4)}
\def\undefined{{\sf undefined}}
\def\PP{{\sf PP}}
\def\TP{{\sf TP}}
\def\PT{{\sf PT}}
\def\TT{{\sf TT}}
\def\PTX{{\sf X}}
\def\Star{{\ast}}
\def\Prop{{\tt Prop}}
\def\Type{{\tt Type}}
\def\U{\mathcal{U}}
\def\dom{\mathrm{dom}}
\def\fv{\mathrm{fv}}
\def\app{{\sf app}}
\def\lam{{\sf lam}}
\def\prd{{\sf prod}}
\def\CIC{CIC}
\def\CCw{CC$^\omega$}

\let\ECC\CCw
\def\lis{{\sf list}}
\def\CCwp{CC$^\omega_\lis$}
\def\nil{{\sf nil}}
\def\cons{{\sf cons}}
\def\listind{{\sf list\_ind}}
\def\listrec{{\sf list\_rec}}

\title{An Intuitionistic Set-theoretical Model of Fully Dependent \CCw}
\date{}
\author{Masahiro Sato, Jacques Garrigue}

\begin{document}
\maketitle

\begin{abstract}
  Werner's set-theoretical model is one of the simplest models of \CIC.
  It combines a functional view of predicative universes with a collapsed view of the impredicative sort `$\Prop$'.
  However this model of $\Prop$ is so coarse that the principle of excluded middle $P \lor \neg P$ holds.
  Following our previous work~\cite{own}, we interpret $\Prop$ into a topological space (a special case of Heyting algebra) to make the model more intuitionistic without sacrificing simplicity.
  We improve on that work by providing a full interpretation of dependent product types, using Alexandroff spaces.
  We also extend our approach to inductive types by adding support for lists.
\end{abstract}

\section{Introduction}\label{introduction}
There are various models of type theory.
Werner's set-theoretical model~\cite{SetsInTypes} provides an intuitive model of \CIC.
It combines a functional view of predicative universes with a collapsed view of the impredicative sort $\Prop$.
However this model of $\Prop$ is so coarse that the principle of excluded middle $P \lor \neg P$ holds in it.

In this paper, we construct a set-theoretical model of \CCw{} in which
the principle of excluded middle does not hold, making it closer to
completeness.

CC (the Calculus of Constructions~\cite{CC}) is a pure type system~\cite{pure_type_system} with two sorts, impredicative $\Star$ and predicative $\Box$.
\CCw{} replaces $\Box$ by a cumulative hierarchy of predicative sorts $\Type_i$.
CIC (the Calculus of Inductive Constructions) adds inductive types to \CCw.

In~\cite{SetsInTypes}, Werner provides a remarkably simple model of CIC.
In this model, $\lambda x : A.t$ is interpreted by a set-theoretical function for predicative sorts.
Yet such a simple approach is known to fail for impredicative sorts as it runs afoul of Reynolds' paradox~\cite{model_not_set}.
Therefore, the model for $\Prop$ is two-valued. 
Hence the principle of excluded middle is valid in this model, making
it classical.
Later, Miquel and Werner \cite{not_simple} have shown that proving the
soundness of this model was not as easy as it seems, but this does not
change the simplicity of the model itself.
This simple approach is to be contrasted with Luo's model of
ECC (\CCw{} extended with strong sums $\Sigma x : A.B$)
which uses $\omega$-sets~\cite{Luo}, syntactic models based on
combinatory logic~\cite{CombinatoryModel, CombinatoryModel2},
or more recent models such as categorical models~\cite{CategoryType, SemanticsOfTypeTheory} or models based on homotopy theory~\cite{HoTT}.
This is the drawback of simplicity: while Werner's approach avoids many complications of more precise models, it is at times counter-intuitive, as it completely ignores the intuitionistic aspect of \CCw.

Our goal has been to recover the intuitionistic part of \CCw{} without increasing the complexity of the model.
Barras~\cite{SetsInCoq} provided a first way to do it, by interpreting
\CCw{} in IZF (intuitionistic Zermelo-Fraenkel set theory~\cite{CZF}) rather than ZF.
While this is an interesting result, and the fact it is backed by a
fully formalized proof is very impressive, this requires one to work
in the radically different world of IZF, where it is difficult to express
meta-reasoning about the expressiveness of the language.
For this reason we prefer to stay inside classical set theory ZF, but
we change the interpretation of $\Prop$ to be some topological space.
The open sets of a topological space form a Heyting algebra.
Heyting algebras are used when constructing models of intuitionistic
logic, but usually their elements are not understood as sets.
In our model, proofs shall be interpreted as elements of denotations
of propositions, hence these denotations must be sets, and the order
must be set inclusion.
Using topological spaces solves this problem.

This leaves the question of how to interpret proofs, in a way that
makes the whole interpretation coherent.
In our previous work~\cite{own}, proofs were interpreted by a constant
value, that had to be included in all true propositions.
But this choice was too inflexible to accomodate propositions
parameterized over proofs, which we had to reject. 
While this type of parameterization is rare, it is for
instance required to express proof-irrelevance as a proposition.
In this paper, we are able to lift this restriction by shifting the
interpretation to Alexandroff spaces~\cite{alexandroff}, and making
the interpretation of proofs a function of the context valuation.
Alexandroff spaces act as parameters to the model, their choice
making it more or less precise.
For instance if we use the trivial topological space $(X, \mathcal{O}(X))$ where $X = \{\cdot\}$ is a singleton and $\mathcal{O}(X)=\{\varnothing, X\}$, we obtain a model of classical logic, which is the coarsest one.

Our model is still proof irrelevant, as it does not depend on the
the proof term itself.
As a result, this model does validate some propositions that are not
provable, in particular logical proof irrelevance, hence it does not
reach completeness.
However this is sufficient to exclude many classical propositions such
as the principle of excluded middle $P \lor \neg P$ or the linearity
axiom $(P \rightarrow Q) \lor (Q \rightarrow P)$.

Note that, in this paper, we choose a slightly restricted version of \CCw,
which omits subsumption between universes $\Prop$ and $\Type_i$.
Subsumption between the predicative universes $\Type_i$ poses no problem,
but our model of propositions is too different to allow
subsumption between $\Prop$ and $\Type_i$.
Werner omitted this same subsumption in his exploration of proof
irrelevance~\cite{the_strength_of_type_theories}.

This model can also be extended to inductive types. To demonstrate it,
we define a model of lists, with principles for recursion (in $\Type_0$)
and induction (in $\Prop$), and extend our soundness proof to those.
This is one more step in the direction of a model for the full CIC.

In section \ref{typetheory}, we define the language of the type system \ECC.
In section \ref{model}, we give our set-theoretical interpretation of
\ECC, prove its soundness, and verify that it satisfies proof irrelevance.
In section \ref{application}, we show some applications of this model.
For instance, we show that the excluded middle cannot be derived from the linearity axiom in \ECC.
In section \ref{sec:inductive}, we extend our interpretation to
inductive types.
Finally, we conclude and discuss some future directions.

\section{Typing of \ECC}\label{typetheory}
\subsection{Definition of \ECC}
We define the type system \ECC{} as follows.
The only deviation from the standard presentation is
that our version has no subsumption between $\Prop$ and $\Type_i$.
\begin{definition}[Term]\label{definition_of_term}
  Let $V$ be an infinite set of variables.
  \begin{itemize}
  \item For all $x \in V$, $x$ is a term with free variables $\fv(x)=\{x\}$.
  \item If $t_1$ and $t_2$ are terms, then $t_1 \; t_2$ is a term with
    free variables $\fv(t_1) \cup \fv(t_2)$.
  \item If $t$ and $T$ are terms, and $x \in V$ then,
    $\lambda x : T. t$ is a term with free variables
    $\fv(T) \cup (\fv(t) \setminus \{x\})$. 
  \item If $T_1$ and $T_2$ are terms, and $x \in V$ then
    $\forall x : T_1. T_2$ is a term with free variables $\fv(T_1)
    \cup (\fv(T_2) \setminus \{x\})$.
  \item The symbols $\Prop$ and $\Type_i$ (for $i = 0,1,2,...$) are terms with free
    variables $\varnothing$.
  \end{itemize} 
\end{definition}
$\Prop$ and $\Type_i$ are called {\em sorts}.
$\Prop$ is called the impredicative sort and it represents the type of all propositions.

\begin{definition}[Context]
  \ 
  \begin{itemize}
  \item $[]$ is a context with domain $\dom([]) = \varnothing$.
  \item If $\Gamma$ is a context, and $T$ is a term and $x \in
    V\setminus \dom(\Gamma)$, then $\Gamma ; (x : T)$ is a context
    with domain $\dom(\Gamma)\cup\{x\}$.
  \end{itemize}
\end{definition}

\begin{table}[h]
  \begin{center}
  \caption{Typing rules of \ECC}
  \begin{tabular}{cc}
    \hline \hline
    $\displaystyle {[]} \vdash  \Prop : \Type_0$ & (Axiom-$\Prop$) \\ \\
    $\displaystyle {[]} \vdash  \Type_i : \Type_{i+1}$ & (Axiom-$\Type$) \\ \\
    $\displaystyle\cfrac{\Gamma \vdash t : T \quad \Gamma \vdash A : s \quad x \notin \dom(\Gamma)}{\Gamma ; (x : A) \vdash t : T}$ &
    (Weakening) \\ \\
    $\cfrac{\Gamma \vdash  A : s_1
      \quad \Gamma ; (x : A) \vdash  B : s_2
      \quad (s_1, s_2) \in \{\Prop, \Type_i\}\times\{\Prop, \Type_i\}}
           {\Gamma \vdash  \forall x:A.B : s_2}$
    & (PI-Type) \\ \\
    $\cfrac{\Gamma ; (x : A) \vdash t : B \quad \Gamma \vdash \forall x:A.B : s}
           {\Gamma \vdash \lambda x:A.t : \forall x:A.B}$
    & (Abstract) \\ \\
    $\cfrac{\Gamma \vdash  u : \forall x:A.B \quad \Gamma \vdash  v : A}
           {\Gamma \vdash  u \; v : B[x \backslash v]}$
    & (Apply) \\ \\
    $\cfrac{\Gamma \vdash A : s \quad x \notin \dom(\Gamma)}{\Gamma ; (x : A) \vdash x : A}$ & (Variable) \\ \\
    $\cfrac{\Gamma \vdash t : A \quad \Gamma \vdash B : s \quad A=_{\beta} B}{\Gamma \vdash t: B}$
    & (Beta Equality) \\ \\
    $\cfrac{\Gamma \vdash t : \forall x_1:A_1,\dots,\forall
        x_n:A_n,\Type_i \quad i < j}
      {\Gamma \vdash t: \forall x_1:A_1,\dots,\forall x_n:A_n,\Type_j}$
    & (Subsumption) \\
    \hline \hline
  \end{tabular}
  \label{tab:typing_rule_of_ECC}
  \end{center}
\end{table}

Table \ref{tab:typing_rule_of_ECC} contains the typing rules of \ECC.
The metavariables $s, s_1, s_2$ denote sorts.
In rule (PI-Type), either $s_1 = s_2$ or one of them is $\Prop$.
The equality $=_{\beta}$ denotes {\em beta equality} and $B[x
  \backslash v]$ denotes {\em substitution}.
Here are their definitions.

\begin{definition}[Substitution]\label{substitution_ecc}
  Let $t$ and $v$ be terms and $x$ be a variable.
  The substitution $t[x \backslash v]$, which means $v$ replaces $x$ in $t$, is defined inductively as follows:
  \begin{enumerate}[{(}i{)}]
  \item If $y$ is a variable, then $y[x \backslash v] = \begin{cases}v & (y=x) \\ y & (otherwise),\end{cases}$
  \item $(t_1 t_2)[x \backslash v] = (t_1[x \backslash v]) (t_2[x \backslash v])$,
  \item $(\lambda x' : T.t')[x \backslash v] = \lambda x' : (T[x
    \backslash v]). t'[x \backslash v]$ \\
    \hfill when $x'\notin\fv(v)\cup\{x\}$,
  \item $(\forall x' : T_1 . T_2)[x \backslash v] = \forall x' :
    (T_1[x \backslash v]) . (T_2[x \backslash v])$ \\
    \hfill when $x'\notin\fv(v)\cup\{x\}$,
  \item $s[x \backslash v] = s$ where $s$ is a sort.
  \end{enumerate}
\end{definition}

\begin{definition}[Beta Equality]\label{beta_ecc}
  Let $=_{\beta}$ be the smallest equivalence relation such that the following conditions hold.
  \begin{enumerate}[(i)]
  \item $(\lambda x : A. t) \; a =_{\beta} t[x \backslash a]$.
  \item If $t_1 =_{\beta} t_1'$ and $t_2 =_{\beta} t_2'$, then $t_1 t_2 =_{\beta} t_1' t_2'$.
  \item If $t =_{\beta} t'$ and $A =_{\beta} A'$, then $\lambda x : A. t =_{\beta} \lambda x : A' t'$.
  \item If $A =_{\beta} A'$ and $B =_{\beta} B'$, then $\forall x : A. B =_{\beta} \forall x : A' B'$.
  \end{enumerate}
\end{definition}

Now that we have defined \ECC 's terms and typing rules,
we show the following three lemmas that will be used in proofs.
They can be proved by induction over the typing rules above.

\begin{lemma}[Uniqueness of Typing]\label{uniqueness_of_type}
  If $\Gamma \vdash t : A$ and $\Gamma \vdash t : B$ are derivable,
  then either $A =_{\beta} B$, or $A =_\beta \forall x_1:A_1,\dots,\forall
  x_n:A_n,\Type_i$ and $B =_\beta \forall x_1:A_1,\dots,\forall
  x_n:A_n,\Type_j$.
\end{lemma}

\begin{lemma}[Substitution]\label{substitution_lemma}
  If $\Gamma \vdash u : U$ and $\Gamma ; (x : U) ; \Delta \vdash t : T$ are derivable then $\Gamma ; \Delta[x \backslash u] \vdash t[x \backslash u] : T[x \backslash u]$ is also derivable.
\end{lemma}

\begin{lemma}[Extended Weakening]
  If $\Gamma_1 ; \Gamma_2 \vdash t : T$ is derivable, then $\Gamma_1 ; \Delta ; \Gamma_2 \vdash t : T$ is also derivable when $\Gamma_1 ; \Delta ; \Gamma_2$ is well-formed, i.e. when $\Gamma_1 ; \Delta ; \Gamma_2 \vdash \Type_i : \Type_{i+1}$ is derivable.
\end{lemma}

\subsection{Propositional terms and proof terms}
In \ECC, propositions are types that belong to the impredicative sort $\Prop$,
and proofs are terms of types that represent propositions.
Next, we give a definition of propositions and proofs through syntactic derivability.
Rather than introducing an explicitly sorted type system like in \cite{not_simple}, we will prove that these definitions are stable under substitution, weakening, and reduction, so that we can safely use them when defining our interpretation.
\begin{definition}
  \
  \begin{enumerate}
  \item {\em Propositional Term} \\
    A term $P$ is called a propositional term for $\Gamma$ iff $\Gamma \vdash P : \Prop$ is derivable.
  \item {\em Proof Term} \\
    A term $p$ is called a proof term for $\Gamma$ iff $\Gamma
    \vdash p : P$ is derivable for some $P$ that is a propositional
    term for $\Gamma$.
    $P$ is then called a provable propositional term for $\Gamma$.
  \end{enumerate}
\end{definition}

\begin{lemma}[Proof and propositional terms]\label{proposition_proof_invariant}
  \
  \begin{enumerate}[{(}i{)}]
  \item We assume that $P_1$ and $P_2$ are well typed under the same context $\Gamma$.
    If $P_1$ is a propositional term for $\Gamma$ and $P_1 =_\beta P_2$, then $P_2$ is also a propositional term for $\Gamma$. \label{enum:proposition_invariant_beta}
  \item We assume that $p_1$ and $p_2$ are well typed under the same context $\Gamma$.
    If $p_1$ is a proof term for $\Gamma$ and $p_1 =_\beta p_2$, then $p_2$ is also a proof term for $\Gamma$.
  \item We assume that $\Gamma \vdash u : \forall x : A.B$ and $\Gamma \vdash v : A$ are derivable. If $u$ is a proof term for $\Gamma$, then $u \; v$ is also a proof term for $\Gamma$.
  \item If $t$ is a proof term for $\Gamma ; (x : A)$ and $\lambda x : A.t$ is well typed under $\Gamma$, then $\lambda x : A.t$ is also a proof term for $\Gamma$.
  \item If $t$ is a proof term for $\Gamma$, then there does not exist a term $T$ such that $\Gamma \vdash t : T$ and $\Gamma \vdash T : \Type_i$ are both derivable.
  \end{enumerate}
\end{lemma}

Proof terms and propositional terms are preserved under substitution.
The following lemma expresses this fact.


\begin{lemma}\label{type_of_type}
  If $\Gamma \vdash t : T$ is derivable, then $\Gamma \vdash T : s$ for some sort $s$.
\end{lemma}
\begin{lemma}\label{substitution_in_proof}
  We assume that $\Gamma \vdash u : U$ is derivable and $p$ is well typed under $\Gamma ; (x : U) ; \Delta$.
  \begin{enumerate}[{(}i{)}]
  \item If $p$ is a proof (resp. propositional) term for the context $\Gamma ; (x : U) ; \Delta$, then $p[x \backslash u]$ is a proof (resp. propositional) term for the context $\Gamma; \Delta[x \backslash u]$.
  \item If $p$ is not a proof term for the context $\Gamma ; (x : U) ; \Delta$, then $p[x \backslash u]$ is not a proof term for the context $\Gamma ; \Delta[x \backslash u]$.
  \end{enumerate}
\end{lemma}
\begin{proof}
  (i) is clear by Lemma~\ref{substitution_lemma}.
  We will show (ii).
  Since $p$ is well typed, there esists a type $T$ such that
  \[ \Gamma ; (x : U) ; \Delta \vdash p : T \]
  and by Lemma~\ref{type_of_type} there exists a sort $s$ such that
  \[ \Gamma ; (x : U) ; \Delta \vdash T : s \]
  Since $p$ is not a proof term for the context $\Gamma ; (x : U) ; \Delta$,
  we have that $s \neq \Prop$, and as a result there exists an index $i$ such
  that $s = \Type_i$.
  Hence by Lemma~\ref{substitution_lemma},
  \begin{eqnarray*}
    \Gamma ; \Delta[x \backslash u] &\vdash& p[x \backslash u] : T[x \backslash u] \\
    \Gamma ; \Delta[x \backslash u] &\vdash& T[x \backslash u] : \Type_i
  \end{eqnarray*}
  hold.
  If there exists a term $P$ such that
  \begin{eqnarray*}
    \Gamma ; \Delta[x \backslash u] &\vdash& p[x \backslash u] : P \\
    \Gamma ; \Delta[x \backslash u] &\vdash& P : \Prop,
  \end{eqnarray*}
  it implies a contradiction by Lemma~\ref{proposition_proof_invariant} (v).
\end{proof}

Note that the fact that $P$ is not a propositional term for $\Gamma ; (x : U) ; \Delta$ does not imply that $P[x \backslash u]$ is not a propositional term for $\Gamma ; \Delta[x \backslash u]$ in general.
Here is a counterexample.
\begin{eqnarray*}
  \Gamma ; (U : \Type_i) ; ( P : U ) &\vdash& P : U \\
  \Gamma &\vdash& \Prop : \Type_i
\end{eqnarray*}
In this case, $P$ is not a propositional term.
However $P[U \backslash \Prop] = P$ is a propositional term under $\Gamma ; (P : \Prop)$.

\begin{lemma}\label{weakening_in_proof}
  We assume that $p$ is well typed under $\Gamma_1 ; \Gamma_2$ and $\Gamma_1 ; \Delta ; \Gamma_2$.
  $p$ is a proof (resp. propositional) term for the context $\Gamma_1 ; \Gamma_2$ if and only if $p$ is a proof (resp. propositional) term for the context $\Gamma_1; \Delta; \Gamma_2$.
\end{lemma}


The function $\PD{\Gamma}{x}{A}{B}$ maps two types into the string symbols $\{\PP, \TP, \PT, \TT\}$.
Its goal is to discriminate cases of $\forall x : A.B$ to give them different interpretations.
\begin{definition}[Product Type]
  We assume that $\Gamma \vdash A : s_1$ and $\Gamma ; (x : A) \vdash B : s_2$ are derivable where $s_1$, $s_2$ are sorts.
  \begin{equation*}
    \PD{\Gamma}{x}{A}{B} := 
    \begin{cases}
      \PP & (s_1, s_2) = (\Prop, \Prop) \\
      \TP & (s_1, s_2) = (\Type_i, \Prop) \\
      \PT & (s_1, s_2) = (\Prop, \Type_i) \\
      \TT & (s_1, s_2) = (\Type_i, \Type_j)
    \end{cases}
  \end{equation*}
\end{definition}
Again, $\PD{\Gamma}{x}{A}{B}$ is stable under substitution and weakening.
\begin{lemma}\label{product_type_substitution_weakening}
  \
  \begin{enumerate}[{(}i{)}]
  \item If $A$ and $B$ are typable under $\Gamma; (x : U) ; \Delta$ and $\Gamma \vdash u : U$ is derivable, then $\PD{(\Gamma; (x : U) ; \Delta)}{a}{A}{B} = \PD{(\Gamma; \Delta[x \backslash u])}{a}{A[x \backslash u]}{B[x \backslash u]}$ holds. \label{enum:pd_substitution}
  \item If $A$ and $B$ are typable under $\Gamma_1 ; \Gamma_2$ and $\Gamma_1; \Delta ; \Gamma_2$, then $\PD{(\Gamma_1 ; \Delta ; \Gamma_2)}{a}{A}{B} = \PD{(\Gamma_1 ; \Gamma_2)}{a}{A}{B}$ holds. \label{enum:pd_weakening}
  \end{enumerate}
\end{lemma}


\begin{proof}
  \begin{enumerate}[{(}i{)}]
  \item When $\PD{\Gamma; (x : U) ; \Delta}{a}{A}{B} = \PP$, $A$ is a proposition for $(\Gamma ; (x : U) ; \Delta)$ and $B$ is a proposition for $(\Gamma ; (x : U) ; \Delta ; (a : A))$.
    By Lemma~\ref{substitution_in_proof}, $A[x \backslash u]$ is a proposition for $(\Gamma ; \Delta[x \backslash u])$ and $B[x \backslash u]$ is also a proposition for $(\Gamma ; \Delta[x \backslash u] ; (a : A[x \backslash u]))$.
    Hence the statement holds in this case.
    When $\PD{\Gamma ; (x : U) ; \Delta}{a}{A}{B} = \TP$, $\Gamma ; \Delta[x \backslash u] \vdash A[x \backslash u] : \Type_i$ is derivable.
    The remaining case is similar.
  \item It is clearly proved by applying the result of $($\ref{enum:pd_substitution}$)$ in this lemma, since variables in $\Delta$ do not appear in $\Gamma_2$ and terms $A$ and $B$.
  \end{enumerate}
\end{proof}

\subsection{Logical symbols}

Lastly, here are some notations allowing to use other logical symbols~\cite{lambda_type}.
We shall use them to prove the adequacy of our model with respect to intuitionistic logic.

\begin{definition}\label{logical_symbol}
  \begin{eqnarray*}
    A \rightarrow B &:=& \forall x : A. B \hspace{8ex} (\text{when } x \notin fv(B)), \\
    \bot &:=& \forall P : \Prop. P, \\
    \neg A &:=& A \rightarrow \bot, \\
    A \land B &:=& \forall P : \Prop. (A \rightarrow B \rightarrow P) \rightarrow P, \\
    A \lor B &:=& \forall P : \Prop. (A \rightarrow P) \rightarrow (B \rightarrow P) \rightarrow P, \\
    \exists x:A.Q &:=& \forall P : \Prop. (\forall x : A. (Q \rightarrow P)) \rightarrow P, \\
    A \leftrightarrow B &:=& (A \rightarrow B) \land (B \rightarrow A), \\
    x =_A y &:=& \forall Q : (A \rightarrow \Prop). Q \; x \leftrightarrow Q \; y.
  \end{eqnarray*}
\end{definition}

\section{Interpretation}\label{model}
\subsection{Preparation of the interpretation}

\subsubsection{Heyting algebras}
Several interpretations of type theory have been proposed such as using $\omega$-sets~\cite{Luo} or coherent spaces~\cite{PAT}.
In this paper, we use {\em Heyting algebras}~\cite{MacLane,IntuitionisticLogic} for propositions.
Heyting algebras provide models of intuitionistic logic.
The open sets of a topological space can be given the structure of a Heyting algebra (see Lemma~\ref{top_forms_hey}), and as such provide models of intuitionistic logic too~\cite{IntuitionisticLogic}.
We give a definition of lattice and Heyting algebra as follows.
\begin{definition}[Lattices and Heyting algebras]
  Let $(A,\le)$ be a partially ordered set (i.e. reflexive, antisymmetric, and transitive).
  $(A,\le)$ is called a {\em Lattice} when any two elements $a$ and
  $b$ of $A$ have a supremum `$a \sqcup b$' and an infimum `$a \sqcap
  b$', which are called join and meet\footnote{
    We use the lattice operation symbols join `$\sqcup$' and meet
    `$\sqcap$' instead of `$\lor$' and `$\land$', since we use the
    latter as logical symbols.
  }.
  A lattice is also called a {\em complete lattice} if every subset
  $S$ of $A$ has a supremum `$\bigsqcup S$' and an infimun `$\bigsqcap S$'.
  A complete lattice has a minimum element $\mathbb{O} := \bigsqcup \varnothing$ and a maximum element $\mathbb{I} := \bigsqcap \varnothing$.
  If a (complete) lattice has an $exponential$ $operator$ $a^b$ such that
  \begin{equation*}
    x \le z^y \Leftrightarrow x \sqcap y \le z
  \end{equation*}
  holds, then we call it a (complete) {\em Heyting Algebra}.
\end{definition}
The following lemma shows that topological spaces are both
Heyting algebras and complete lattices.

\begin{lemma}\label{top_forms_hey}
  Any topological space $(X, \mathcal{O}(X))$ is a complete Heyting algebra.
\end{lemma}
\begin{proof}
  Let $a \le b$ be $a \subset b$, and define each operation as follows:
  \begin{eqnarray*}
    \mathbb{I} &:=& X, \\
    \mathbb{O} &:=& \varnothing, \\
    \bigsqcup S &:=& \bigcup S, \\
    \bigsqcap S &:=& \bigsqcup\{t \mid \forall s \in S, t \leq s\} =
                     \biggl(\bigcap S \biggr)^\circ \\
    && \hspace{8ex} (where \; A^\circ \; is \; the \; interior \; of \; A), \\
    b^a &:=& \bigsqcup\{t \mid t \sqcap a \leq b\}.
  \end{eqnarray*}
\end{proof}

The following lemma states well known properties of complete Heyting algebras.
\begin{lemma}\label{heyting_conditions}
  Let $(A,\leq)$ be a complete Heyting algebra. Then the following conditions hold.
  \begin{eqnarray}
    (x^b)^a &=& x^{a \sqcap b}, \label{eq:powerprod}\\
    \bigsqcap \{t^{t^a} \; | \; t \in A\} &=& a, \label{eq:meetpower}\\
    x^a \sqcap x^b &=& x^{a \sqcup b}, \label{eq:prodpoweror}\\
    \bigsqcap \{a^t \; | \; t \in S\} &=& a^{\bigsqcup S}, \label{eq:meetpoweror}\\
    \bigsqcap \varnothing &=& 1, \label{eq:whole}\\
    x^1 &=& x, \label{eq:exp_unit_1}\\
    y &\leq& y^x, \label{eq:powerle}\\
    x \leq y &\Rightarrow& y^x = 1 \label{eq:le_whole}\\
    y \leq x \mbox{ and } x \not\leq y &\Rightarrow& y^x = y \\
    x \sqcap y^x &\leq& y, \label{eq:impl}\\
    x^y \sqcap y ^ x = 1 &\Rightarrow& x = y, \label{eq:power_eqcond}\\
    \bigsqcap S = 1 &\Rightarrow& \forall a \in S, a = 1. \label{eq:topcond} \\
    \bigl( \bigsqcap \{f(t) \mid t \in A \} \bigr)^x &=& \bigsqcap \{f(t)^x \mid t \in A \} \label{eq:to_inner_exp}
  \end{eqnarray}
\end{lemma}

\subsubsection{Alexandroff spaces}
In our interpretation, a proof term is interpreted into an element of an open set.
In our previous work~\cite{own}, all proof terms were interpreted into a single point, the {\em reference point}.
Soundness then required this reference point to be included in the interpretation of all propositions in the context, which forced us to restrict the type system.
In this paper, we make the  interpretation of proofs a function of the context, which allows us to overcome this restriction.

As a first step, we discuss Alexandroff spaces~\cite{alexandroff}.
\begin{definition}[Alexandroff Space]
  A topological space $(X, \mathcal{O}(X))$ is an Alexandroff space iff the intersection of any tribe of open set is also an open set, i.e.
  \begin{equation*}
    \bigcap S \in \mathcal{O}(X) \quad \mathrm{for \; any \;} S \subset \mathcal{O}(X)
  \end{equation*}
\end{definition}

The definition of Alexandroff space can also be given by the following equivalent statement.

\begin{lemma}[Minimal Neighborhood]
  A topological space $(X, \mathcal{O}(X))$ is an Alexandroff space iff any point has
  a minimal neighborhood.  The minimal neighborhood of the point $x$ is denoted by $\downarrow x$.
\end{lemma}

These are the basic definitions for Alexandroff spaces.
However, to prove our soundness theorem later, we need more conditions.
We state those as {\em well behaved} Alexandroff spaces.

\begin{definition}[Well Behaved Alexandroff Space]\label{wbats}
  An Alexandroff space $(X, \mathcal{O}(X))$ is well behaved if the following conditions hold.
  \begin{itemize}
  \item For any finite subset $\{t_1, t_2, \cdots, t_n\}$ of $X$, we can choose a point $t \in X$ such that
    \begin{equation*}
      \downarrow t_1 \; \cap \; \downarrow t_2 \; \cap \cdots \cap \; \downarrow t_n = \; \downarrow t
    \end{equation*}
    holds.
    We write such a point $t$ as $\inf\{t_1, t_2, \cdots , t_n\}$.
  \item There exists an element $\bot_X \in X$ such that any inhabited open set contains it, i.e.
    \begin{equation*}
      \forall O \in \mathcal{O}(X), O \mathrm{\; is \; inhabited} \Rightarrow \bot_X \in O.
    \end{equation*}
  \end{itemize}
\end{definition}

To clarify the use of the notation of the minimal neighborhood $\downarrow x$ and $\bot_X$, let us discuss a preordered (i.e. reflexivity and transitivity hold) set generated from an Alexandroff space.
Let $\leq$ be the relation on $X$ defined as follows.
\begin{equation*}
  a \leq b \quad :\Leftrightarrow \quad \forall O \in \mathcal{O}(X), b \in O \Rightarrow a \in O
\end{equation*}
The relation $\leq$ is a preorder.
Moreover, if this Alexandroff space is a $T_0$ space, then the
generated preorder $(X, \leq)$ becomes an order (the
antisymmetry condition holds).
If the relation $(X, \leq)$ generated from an Alexandroff space forms an ordered set then the followings holds.
\begin{eqnarray*}
  \downarrow x &=& \{ t \in X \mid t \leq x\} \\
  \bot_X &=& \min X
\end{eqnarray*}

Using an ordered Alexandroff space for $X$ allows us to give multiple interpretations of proofs in the typing context, whereas in our previous work~\cite{own} we used a fixed point $p \in X$. This fixed point was required to satisfy a {\em point condition}, which was no other than the existence of a minimal neighborhood, satisfied by every point in an Alexandroff space.

\subsubsection{Dependent function and Universes}

\begin{definition}[Dependent Function]
  Let $A$ be a set, and $B(a)$ be a set with parameter $a \in A$.
  We define the set of dependent functions as follows
   \begin{equation*}
     \prod_{a \in A}B(a) := \{f \subset \coprod_{a \in A}B(a) \; | \; \forall a \in A, \exists ! b \in B(a), (a, b) \in f\}
   \end{equation*}
   that is the set of functions whose graphs are included in
   \begin{equation*}
     \coprod_{a \in A} B(a) := \{(x,y) \in A \times \bigcup_{a \in A}B(a) \; | \; y \in B(x) \}.
  \end{equation*}
\end{definition}

Next, we introduce Grothendieck universes, which are closed under dependent-function construction, and which we will use to interprete the sort $\Type_i$.
\begin{definition}[Grothendieck Universe]
  We define a $i$-th Grothendieck Universe $\U_i$ as
  \begin{equation*}
    \U_i := V_{\lambda_i},
  \end{equation*}
  where a set $V_\alpha$, with an ordinal number $\alpha$, is recursively defined as follows
  \begin{eqnarray*}
    V_0 &=& \varnothing, \\
    V_{\alpha + 1} &=& \mathcal{P}(V_\alpha), \\
    V_\alpha &=& \bigcup_{\beta < \alpha} V_\beta \quad \mbox{(when $\alpha$ is a limit ordinal)},
  \end{eqnarray*}
  and $\lambda_i$ is the $i$-th inaccessible cardinal.
\end{definition}

The class of all universes is well founded for the relation $\in$.
We write $\U_i$ as the $i$-th universe.
Note that $\U_i$ is so large that it cannot be constructed in ZFC without assuming an inaccessible cardinal.
The following lemma is necessary when proving soundness.

\begin{lemma}\label{univ_close}
  The followings hold for any $i$.
  \begin{enumerate}[{(}i{)}]
  \item $A \in \U_i$ implies $A \subset \U_i$.
  \item $A \in \U_i$ and $B_\alpha \in \U_i$ for all $\alpha \in A$ imply $\displaystyle\prod_{\alpha \in A}B_\alpha \in \U_i$. \label{enum:univ_prod}
  \item $x \in \U_i$ and $y \subset x$ imply $y \in \U_i$ \label{enum:univ_transitive}
  \item $\U_i \subset \U_{i+1}$. \label{enum:univ_subset}
  \end{enumerate}
\end{lemma}

\subsection{Interpretation of the judgments}
In this model, a type $T$ is interpreted into a set $\jump{T}$, and a
context $x_1 : T_1 ; x_2 : T_2 ; \cdots ; x_n : T_n$  is interpreted
into a dependent tuple; in particular, when there are no dependent
types in the context, it is a tuple in $\jump{T_1} \times \jump{T_2}
\times \cdots \times \jump{T_n}$.

First, we define the interpretation of application and PI-Type.
The interpretation of application depends on whether the argument is a proof term or not, and may be undefined.
We shall later prove that every time we use it, we actually have $\app_{\Gamma,v}(f,a) = f(a)$.
\def\calA{{\cal A}}
\def\calB{{\cal B}}
\begin{definition}\label{app_prod_int}
  \begin{eqnarray*}
    \app_{\Gamma,v}(f, a) &:=&  \begin{cases}
      f(\bot_X) \\
      \quad ( v \mbox{ is a proof term for } \Gamma \\
      \quad\quad \mbox{and } f \mbox{ is a function whose domain contains } a \mbox{ and } \bot_X) \\
      \\
      f(a) \\
      \quad ( v \mbox{ is not a proof term for } \Gamma \\
      \quad\quad \mbox{and } f \mbox{ is a function whose domain contains } a) \\
      \\
      \undefined \\
      \quad (\mbox{otherwise})
    \end{cases} \\
    \prd_\PTX(\calA, \{\calB(\alpha)\}_{\alpha \in \calA}) &:=&
    \begin{cases}
      \biggl(\bigsqcap\{\calB(\alpha) \mid \alpha\in \calA\}\biggr)^\calA \\
      \quad (\mbox{when } \PTX = \PP) \\
      \\
      \bigsqcap\{\calB(\alpha) \mid \alpha \in \calA \} \\
      \quad (\mbox{when } \PTX = \TP) \\
      \\
      \{f \in \prod_{\alpha \in \calA} \calB(\alpha) \; | \; f \mbox{ is a constant function}\} \\
      \quad (\mbox{when } \PTX = \PT) \\
      \\
      \prod_{\alpha \in \calA} \calB(\alpha) \\
      \quad (\mbox{when } \PTX = \TT) \\
    \end{cases}
  \end{eqnarray*}
\end{definition}

Now, we define the (partial) interpretations of contexts $\jump{\mathrm{-}}$ and judgments $\jump{\mathrm{-} \vdash \mathrm{-}}$. The former is by induction on the length of the context, and the latter by induction on the structure of terms. 
Note that the interpretation of judgments does not rely on the interpretation of contexts.

\begin{definition}[interpretation]\label{interpretation}
  Let $(X, \mathcal{O}(X)) \in \mathcal{U}_0$ be a well behaved Alexandroff space.
  \begin{enumerate}[{(}i{)}]
  \item Definition of the interpretation of a context $\jump{\Gamma}$ \\
    \begin{eqnarray*}
      \jump{[]} &:=& \{()\} \\
      \jump{\Gamma ; (x:A)} &:=&
      \{(\gamma, \alpha) \mid \gamma \in \jump{\Gamma}  \; \mathrm{and} \; \alpha \in \jump{\Gamma \vdash A}(\gamma) \} \\
      &=& \coprod_{\gamma \in \jump{\Gamma}} \jump{\Gamma \vdash A}(\gamma)
    \end{eqnarray*}
    where $()$ represents the empty sequence.
  \item Definition of the interpretation of a judgment $\jump{\Gamma \vdash t}$ \\
    If $t$ is a proof term for $\Gamma = (x_1 : T_1) ; \cdots ; (x_n : T_n)$, then
    \begin{equation*}
      \jump{\Gamma \vdash t}(\gamma) := \floor{\gamma}
    \end{equation*}
    where
    \begin{equation*}
      \floor{\gamma_1, \gamma_2, \cdots, \gamma_n} := \;  \inf \{ \gamma_i \mid x_i \mbox{ is a proof under } \Gamma  \}.
    \end{equation*}
    Otherwise, if $\Gamma \vdash t : T$ is derivable and $T$ is not a
    proposition for $\Gamma$, then
    \begin{align*}
      \begin{array}{rcl}
          \jump{\Gamma \vdash \Type_i}(\gamma) &:=& \U_i \\
          \jump{\Gamma \vdash \Prop}(\gamma) &:=& \mathcal{O}(X) \\
          \multicolumn{3}{l}{
            \jump{(x_1 : T_1) ; \cdots ; (x_n : T_n) \vdash x_i}(\gamma_1, \cdots, \gamma_n) := \gamma_i
          } \\
          \jump{\Gamma \vdash \forall x:A.B}(\gamma) &:=& \prd_\PTX \left(\mathcal{A}, \{\mathcal{B}(\alpha)\}_{\alpha \in \mathcal{A}} \right) \\
          && \quad \mbox{where} \\
          && \qquad \begin{cases}
            \PTX := \PD{\Gamma}{x}{A}{B} \\
            \mathcal{A} := \jump{\Gamma \vdash A}(\gamma) \\
            \mathcal{B}(\alpha) := \jump{\Gamma ; (x : A) \vdash B}(\gamma, \alpha)
          \end{cases} \\
          \jump{\Gamma \vdash \lambda x:A.t}(\gamma) &:=& \bigl\{\bigl(\alpha, \jump{\Gamma ; (x : A) \vdash t}(\gamma, \alpha)\bigr) \; | \; \alpha \in \jump{\Gamma \vdash A}(\gamma) \Bigr\} \\
          \jump{\Gamma \vdash u \; v}(\gamma) &:=& \app_{\Gamma, v}(\jump{\Gamma \vdash u}(\gamma), \jump{\Gamma \vdash v}(\gamma))
      \end{array}
    \end{align*}
  \end{enumerate}
  For simplicity, we write $\jump{T}$ for $\jump{[]\vdash T}()$, when the context is empty.
\end{definition}


When defined, the interpretation of a context $\jump{\Gamma}$ is a set
of sequences $\gamma$ whose length is the length of $\Gamma$, and
$\jump{\Gamma \vdash t}$ is a function whose domain is
$\jump{\Gamma}$, and which returns some set $\jump{\Gamma \vdash t}(\gamma)$
--- soundness will tell us that if $\Gamma \vdash t : T$, then
$\jump{\Gamma \vdash t}(\gamma) \in \jump{\Gamma \vdash T}(\gamma)$.

Concerning Definitions~\ref{app_prod_int} and \ref{interpretation},
most cases are similar to Werner's interpretation, and we explained
$\app_{\Gamma,v}$ above, so we only explain the interpretations of
proof terms and PI-Types $\forall x : A.B$.

The interpretation of a proof term $\jump{\Gamma \vdash p}(\gamma)$ is the minimum element of the set of proof values in $\gamma$.
Since each of these values belong to the interpretations of propositions in the context, which are open sets in our Alexandroff space, this minimum element belongs to all of them. This will allow us to prove that any proof variable belongs to the interpretation of its type, which is key to the soundness theorem.

$\prd_\PTX$ has four cases, according to $\PTX = \PD{\Gamma}{x}{A}{B}$.
When $\PTX = \PP$, we use the Heyting algebra representation of this implication.
If $x$ does not appear in B, the interpretation of $\jump{\Gamma \vdash \forall x : A.B}$ is $\calB^\calA$, which represents the logical implication $A \Rightarrow B$, as will be proved in Lemma~\ref{last_cor}. If $x$ appears in $B$, we still have the same meaning, since $\calB(\alpha)$ does not depend on $\alpha$, as will be proved in Lemma~\ref{interpretation_constant}. This definition also works if $\calA$ is empty, as the empty meet is $X$, and $X^\emptyset$ is $X$ again (the top element of the lattice). In our previous work, $\alpha$ was required to be the (fixed) interpretation of a proof term, meaning that we could not interprete the case where $\calA$ was not empty, but did not contain the reference point used for proof terms. Here we do not have such a problem, as the interpretation of proof terms is a function of the context; thanks to the interpretation with well behaved Alexandroff spaces, there is always a value small enough to serve as proof term.

When $\PTX = \TP$, the interpretation of $\jump{\Gamma \vdash \forall x : A.B}$ represents universal quantification, and again we use the infinite meet operator of the complete Heyting algebra to express it.

When $\PTX = \PT$, the interpretation of $\jump{\Gamma \vdash \forall x : A.B}$ becomes a set theoretical constant function.
Functions whose argument are proofs should be constant functions since our model is proof-irrelevant.

In the last case, when $\PTX = \TT$, the representation becomes a set theoretical dependent function.

As soon as one component is $\undefined$ the whole interpretation is $\undefined$.
Thanks to Corollary~\ref{last_cor} which is a consequence of the Soundness Theorem~\ref{soundness}, $\undefined$ never appears, and implication and application can be defined in a straightforward way.
\begin{corollary}\label{last_cor}
  \
  \begin{itemize}
  \item If $\Gamma \vdash t$ is well typed, then $\jump{\Gamma \vdash t}$ is a total function whose domain is $\jump{\Gamma}$.
  \item If $\PD{\Gamma}{x}{A}{B}=\PP$ and $\jump{\Gamma \vdash A}(\gamma) \neq \varnothing$, then
    \begin{equation*}
      \jump{\Gamma \vdash \forall x : A.B}(\gamma) = {\jump{\Gamma ; (x : A) \vdash B}(\gamma, \alpha)}^{\jump{\Gamma \vdash A}(\gamma)}
    \end{equation*}
    holds for any $\alpha \in \jump{\Gamma \vdash A}(\gamma)$.
  \item If $\PD{\Gamma}{x}{A}{B}=\PP$ and $\jump{\Gamma \vdash A}(\gamma) = \varnothing$, then
    \begin{equation*}
      \jump{\Gamma \vdash \forall x : A.B}(\gamma) = X
    \end{equation*}
    holds.
  \item If $\Gamma \vdash t_1 \; t_2$ is well typed and $t_1$ is not a proof term for $\Gamma$, then $\jump{\Gamma \vdash t_1}(\gamma)$ is a function whose domain contains $\jump{\Gamma \vdash t_2}(\gamma)$ and
    \begin{equation*}
      \jump{\Gamma \vdash t_1 \; t_2}(\gamma) = \jump{\Gamma \vdash t_1}(\gamma)\biggl(\jump{\Gamma \vdash t_2}(\gamma)\biggr)
    \end{equation*}
    holds.
  \end{itemize}
\end{corollary}

\subsection{Soundness}

We can now start our soundness proof with the weakening and
substitution lemmas.
They show that our interpretation is well behaved.

\begin{lemma}[interpretation of weakening]\label{interpretation_composition}
  If $t$ is not a proof term, then the following equation holds
  \begin{equation*}\jump{\Gamma_1 ; \Gamma_2 \vdash t}(\gamma_1, \gamma_2) = \jump{\Gamma_1 ; (x' : A') ; \Gamma_2 \vdash t}(\gamma_1, \alpha', \gamma_2)\end{equation*}
  when both sides are well defined.
\end{lemma}
\begin{proof}
  See Appendix~\ref{interpretation_weakening_appendix}.
\end{proof}

Our substitution lemma is similar to those in~\cite{SetsInTypes} and~\cite{not_simple}.
\begin{lemma}[interpretation of substitution]\label{substitution_interpretation}
  We assume $\Gamma \vdash u : U$ is derivable.
  If $\Gamma ; (x : U) ; \Delta$ is well formed and
  \begin{equation*}
    (\gamma, \jump{\Gamma \vdash u}(\gamma), \delta) \in \jump{\Gamma; (x : U) ; \Delta} 
  \end{equation*}
  holds (with all interpretations defined), then
  \begin{equation*}
    (\gamma, \delta) \in \jump{\Gamma; \Delta[x \backslash u]}
  \end{equation*}
  holds.
  Moreover, in
  \begin{equation*}
    \jump{\Gamma; (x:U) ; \Delta \vdash t}(\gamma, \jump{\Gamma \vdash
       u}(\gamma),\delta)
    = \jump{\Gamma ; \Delta[x \backslash u] \vdash t[x
       \backslash u]}(\gamma,\delta) 
  \end{equation*}
  the right hand side is defined whenever the left hand side is,
  and the equation holds for all $t$ and $T$ such that
  $\Gamma ; (x : U) ; \Delta \vdash t : T$ is derivable.
\end{lemma}
\begin{proof}
  See Appendix~\ref{interpretation_substitution_appendix}.
\end{proof}

While propositions can be interpreted by sets with multiple values, our
interpretation is still proof-irrelevant, as the interpretation of
terms of sort \Type{} does not depend on parameters of sort \Prop.
This simplifies the proof of the next lemma.
\begin{lemma}[semantic proof irrelevance]\label{interpretation_constant}
  We assume that $A'$ is a propositional term for $\Gamma$ and $t$ is not a proof term under $\Gamma ; (x' : A') ; \Delta$.
  If
  \begin{eqnarray*}
    (\gamma, p_1, \delta) &\in& \jump{\Gamma ; (x' : A') ; \Delta} \\
    (\gamma, p_2, \delta) &\in& \jump{\Gamma ; (x' : A') ; \Delta}
  \end{eqnarray*}
  hold, then
  \begin{equation*}
    \jump{\Gamma ; (x' : A') ; \Delta \vdash t : T}(\gamma, p_1, \delta) = \jump{\Gamma ; (x' : A') ; \Delta \vdash t : T}(\gamma, p_2, \delta)
    \end{equation*}
  holds.
\end{lemma}
\begin{proof}
  See Appendix~\ref{interpretation_constant_appendix}.
\end{proof}

\begin{theorem}[soundness of beta equality]\label{soundness_of_beta_equality}
  If $t_1 =_\beta t_2$, and $\Gamma \vdash t_1 : T, \Gamma \vdash t_2 : T$ are derivable, then $\jump{\Gamma \vdash t_1}(\gamma) = \jump{\Gamma \vdash t_2}(\gamma)$ when both sides are well defined.
\end{theorem}
\begin{proof}
  If $t_1$ is a proof term, then $t_2$ is also a proof term by Lemma~\ref{proposition_proof_invariant}, hence the statement holds.
  If not, it is sufficient to only prove that $\jump{\Gamma \vdash (\lambda x :U. t) \; u}(\gamma) = \jump{\Gamma \vdash t[x \backslash u]}(\gamma)$ holds.
  If $(\lambda x : U. t) u$ is well typed under $\Gamma$, then $\Gamma \vdash u : U$ is derivable.
  If $u$ is not a proof term, then
  \begin{eqnarray*}
    && \jump{\Gamma \vdash (\lambda x : U.t) \; u}(\gamma) \\
    &=& \jump{\Gamma \vdash \lambda x : U.t}(\gamma)\bigl(\jump{\Gamma \vdash u}(\gamma)\bigr) \\
    &=& \jump{\Gamma ; (x : U) \vdash t}(\gamma, \jump{\Gamma \vdash u}(\gamma)) \\
    &=& \jump{\Gamma \vdash t[x \backslash u]}(\gamma)
  \end{eqnarray*}
  holds by Lemma~\ref{substitution_interpretation}.
  If $u$ is a proof term, then $\jump{\Gamma \vdash \lambda x : U. t}(\gamma)$ is a function whose domain contains $\jump{\Gamma \vdash u}(\gamma)$ by definition of the interpretation.
  Therefore $\jump{\Gamma ; (x : U) \vdash t}(\gamma, \jump{\Gamma \vdash u}(\gamma))$ is also well defined.
  Hence
  \begin{eqnarray*}
    && \jump{\Gamma \vdash (\lambda x : U.t) \; u}(\gamma) \\
    &=& \jump{\Gamma \vdash \lambda x : U.t}(\gamma)(\bot_X) \\
    &=& \jump{\Gamma ; (x : U) \vdash t}(\gamma, \bot_X) \\
    &=& \jump{\Gamma ; (x : U) \vdash t}(\gamma, \jump{\Gamma \vdash u}(\gamma)) \\
    &=& \jump{\Gamma \vdash t[x \backslash u]}(\gamma)
  \end{eqnarray*}
  holds by Lemma~\ref{substitution_interpretation} and \ref{interpretation_constant}.
  Hence, the statement holds.
\end{proof}

We are now ready to prove the soundness of this type system.

\begin{theorem}[soundness]\label{soundness}
  We assume $\gamma \in \jump{\Gamma}$.
  If $\Gamma \vdash t : T$ is derivable, then $\jump{\Gamma \vdash t}(\gamma) \in \jump{\Gamma \vdash T}(\gamma)$.
\end{theorem}
\begin{proof}
  See Appendix~\ref{soundness_ind}.
\end{proof}

\subsection{Interpretation of logical synbols}

We also prove the following theorem about the interpretation of logical symbols in definition \ref{logical_symbol}.
It demonstrates the logical adequacy of the interpretation.
\begin{theorem}[interpretation of logical symbols]\label{logical_symbol_int}
  \
  \begin{enumerate}[{(}i{)}]
  \item $\jump{\Gamma \vdash \bot}(\gamma) = \varnothing$
  \item $\jump{\Gamma \vdash A \land B}(\gamma) = (\jump{\Gamma \vdash A}(\gamma)) \sqcap (\jump{\Gamma \vdash B}(\gamma))$
  \item $\jump{\Gamma \vdash A \lor B}(\gamma) = (\jump{\Gamma \vdash A}(\gamma)) \sqcup (\jump{\Gamma \vdash B}(\gamma))$
  \item When $A$ is a propositional term: \\
    $\jump{\Gamma \vdash \exists x : A. B}(\gamma) = \begin{cases}\jump{\Gamma \vdash A}(\gamma) \sqcap \jump{\Gamma ; x : A \vdash B}(\gamma, \alpha) & (\alpha \in \jump{\Gamma \vdash A}(\gamma)) \\ \varnothing & (\jump{\Gamma \vdash A}(\gamma) = \varnothing)\end{cases}$
  \item When $A$ is not a propositional term: \\
    $\jump{\Gamma \vdash \exists x : A. B}(\gamma) = \bigsqcup_{\alpha \in \jump{\Gamma \vdash A}(\gamma)} \jump{\Gamma ; (x:A) \vdash B}(\gamma, \alpha)$
  \item $\jump{\Gamma \vdash A \leftrightarrow B}(\gamma) = X \Rightarrow \jump{\Gamma \vdash A}(\gamma) = \jump{\Gamma \vdash B}(\gamma)$
  \item $\jump{\Gamma \vdash x =_A y}(\gamma) = X \Leftrightarrow \jump{\Gamma \vdash x}(\gamma) = \jump{\Gamma \vdash y}(\gamma)$ \label{logical_symbol_int:eq}
  \end{enumerate}
\end{theorem}
To prove Theorem~\ref{logical_symbol_int}, one uses Lemmas~\ref{heyting_conditions}, \ref{interpretation_composition} and Corollary~\ref{last_cor}.
For a detailed proof, see~\cite{own}.

\subsection{Interpretation of logical proof irrelevance}
A general form of proof irrelevance, that does not depend on the type
of the result, can be expressed as a logical formula, using
the propositional encoding for equality:
\begin{equation*}
  \vdash \forall P : \Prop. \forall p_1, p_2 : P. p_1 =_P p_2.
\end{equation*}

\begin{proposition}[interpretation of proof irrelevance]
The logical formula for proof irrelevance is valid for any Alexandroff
space.
\end{proposition}
\begin{proof}
We shall prove that for any topological space $(X, \mathcal{O}(X))$,
the interpretation of this formula is $X$.
First note that, for any valuation $\gamma$,
\[\jump{P : \Prop ; p_1 : P ; p_2 : P \vdash
  p_1}(\gamma) = \floor\gamma = \jump{P : \Prop ; p_1 : P ; p_2 : P
  \vdash p_2}(\gamma). \]
By using (\ref{logical_symbol_int:eq}) from Theorem
\ref{logical_symbol_int}, we have
\[ \jump{P : \Prop ; p_1 : P ; p_2 : P \vdash p_1 =_P p_2}(\gamma) = X.\]
As a result,
\begin{eqnarray*}
  && \jump{\forall P : \Prop. \forall p_1, p_2 : P. p_1 =_P p_2} \\
  && \quad = \bigsqcap_{o \in \mathcal{O}(X)} \bigsqcap_{x_1, x_2 \in o} \jump{P : \Prop ; p_1 : P ; p_2 : P \vdash p_1 =_P p_2}(o, x_1, x_2) \\
  && \quad = \bigsqcap_{o \in \mathcal{O}(X)} \bigsqcap_{x_1, x_2 \in
    o} X \\
  && \quad = X.
\end{eqnarray*}
\end{proof}

Note that semantic proof irrelevance and logical proof irrelevance are
quite different. The former is about equality of interpretations of
non-proof terms under different valuations, while the latter uses the
equality of interpretations of proof terms under the same valuation.
As a result, their proofs are independent.

\section{Application}\label{application}
Let us compare Werner's classical model with our intuitionistic model on some simple cases.
\subsection{Classical model}
We start with the simplest case.
Let us consider the trivial topological space, whose base set is the singleton $\{ \varnothing \}$.
\begin{eqnarray*}
  X &:=& \{ \varnothing \} \\
  \mathcal{O}(X) &:=& \{\varnothing, \{\varnothing\}\} = \{0,1\}
\end{eqnarray*}
This topological space is a well behaved Alexandroff space, and coincides with Werner's Model~\cite{SetsInTypes}.
However this model is so coarse that it represents classical logic, since the principle of excluded middle holds.
\[
  \varnothing \in \jump{\forall P : \Prop. P \lor \neg P} = \bigsqcap_{o \in \mathcal{O}(X)} o \lor \neg o = 1.
\]

If we want to be more discriminating, we need more open sets in
$\mathcal{O}(X)$.

\subsection{Models disproving excluded middle}

Now, let us consider the next simplest topological space, which contains another element.
\begin{eqnarray*}
  X &:=& \{\varnothing, \{\varnothing\}\} \\
  \mathcal{O}(X) &:=& \{\varnothing, \{\varnothing\}, \{\varnothing, \{\varnothing\}\} = \{0, 1, 2\}
\end{eqnarray*}

Although this model stays simple, its topological space is fine enough to avoid the principle of excluded middle, since the following statement holds.
\begin{equation*}
  2 \notin \jump{\forall P : \Prop. P \lor \neg P} = 1.
\end{equation*}
This statement is derived by using the following equations.
\[ \neg 0 = 2 \hspace{4ex} \neg 1 = 0 \hspace{4ex} \neg 2 = 0 \]
By our soundness theorem, this proves that the principle of excluded middle cannot be deduced in \ECC.

\begin{table}[t]
  \caption{Value of $y^x$ for $X=\{\varnothing,\{\varnothing\}\}$}
  \begin{center}
    \begin{tabular}{|c||c|c|c|}
      \hline
      $y^x$ & $0$ & $1$ & $2$ \\ \hline \hline
      $0$ & $2$ & $0$ & $0$ \\ \hline
      $1$ & $2$ & $2$ & $1$ \\ \hline
      $2$ & $2$ & $2$ & $2$ \\ \hline
    \end{tabular}
  \end{center}
  \label{tab:valuea}
\end{table}




Yet this model is not fully intutionistic as the linearity axiom $(P \rightarrow Q) \lor (Q \rightarrow P)$ holds,
since we have the following fact by Table \ref{tab:valuea}.

\begin{eqnarray*}
  && \jump{\forall P : \Prop. \forall Q : \Prop. (P \rightarrow Q) \lor (Q \rightarrow P)} \\
  &=& \bigsqcap_{o_1, o_2 \in \mathcal{O}(X)} o_1^{o_2} \lor o_2^{o_1} \\
  &=& 2.
\end{eqnarray*}
This is actually interesting because it shows that we can use this model to prove non trivial facts,
for instance that the excluded middle cannot be deduced from the linearity axiom in \ECC.
Indeed,
\begin{equation*}
  \jump{(\forall P : \Prop. \forall Q : \Prop. (P \rightarrow Q) \lor
    (Q \rightarrow P)) \rightarrow (\forall P : \Prop. P \lor \neg P) } = 1.
\end{equation*}

By our soundness theorem, this equation means that there is no term proving the above implication in \ECC.\par

\begin{table}[t]
  \caption{Value of $y^x$ for $X=\{b, l , r, t\}$}
  \begin{center}
    \begin{tabular}{|c||c|c|c|c|c|c|}
      \hline
      $y^x$ & $\varnothing$ & $\alpha$ & $\beta$ & $\gamma$ & $\delta$ & $X$ \\ \hline \hline
      $\varnothing$ & $X$ & $\varnothing$ & $\varnothing$ & $\varnothing$ & $\varnothing$ & $\varnothing$ \\ \hline
      $\alpha$ & $X$ & $\alpha$ & $\alpha$ & $\alpha$ & $\alpha$ & $\alpha$ \\ \hline
      $\beta$ & $X$ & $X$ & $X$ & $\alpha$ & $\beta$ & $\beta$ \\ \hline
      $\gamma$ & $X$ & $X$ & $\alpha$ & $X$ & $\gamma$ & $\gamma$ \\ \hline
      $\delta$ & $X$ & $X$ & $X$ & $X$ & $X$ & $\delta$ \\ \hline
      $X$ & $X$ & $X$ & $X$ & $X$ & $X$ & $X$ \\ \hline
    \end{tabular}
  \end{center}
  \label{tab_valueb}
\end{table}


By adding more elements we can refine the model further.
Let $(X, \mathcal{O}(X))$ be the Alexandroff space
\begin{eqnarray*}
  X &:=& \{b, l, r, t\} \\
  \mathcal{O}(X) &=& \{\varnothing, \{b\}, \{b, l\}, \{b, r\}, \{b, l, r\}, X\} \\
  &\equiv& \{\varnothing, \alpha, \beta, \gamma, \delta, X\}
\end{eqnarray*}

In this model, $(P \rightarrow Q) \lor (Q \rightarrow P)$ does not hold,
since we have the following fact by Table \ref{tab_valueb}.
\begin{equation*}
  t \notin \jump{\forall P : \Prop. \forall Q : \Prop. (P \rightarrow Q) \lor (Q \rightarrow P)} = \alpha
\end{equation*}

\section{Interpretation of inductive types}\label{sec:inductive}
Until now, we have discussed the interpretation of \CCw.
However, Coq's type system is not \CCw \; but \CIC, which is \CCw \;
extended with (co)inductive types.
In this paper, we do not give a general definition of inductive types, but we present some examples of inductive definitions.
Here, we introduce a new type system \CCwp, which is \CCw \; with the $\lis$ type.

\subsection{Typing Rule of \CCwp}
To construct the new type system \CCwp, we add new terms and typing rules to \CCw.
Here, we give five new terms, \lis, \nil, \cons, \listrec, and \listind, and also give new typing rules for the $\lis$ type in Table \ref{tab:typing_rule_of_ECC_with_list}.

\begin{table}[h]
  \begin{center}
  \caption{New typing rules of \CCwp}
  \begin{tabular}{cc}
    \hline \hline
    $\displaystyle {[]} \vdash \lis : \Type_0 \rightarrow \Type_0$ & (\lis -intro) \\ \\
    $\displaystyle {[]} \vdash \nil : \forall A : \Type_0. \lis \; A$ & (\nil -intro) \\ \\
    $\displaystyle {[]} \vdash \cons : \forall A : \Type_0. A \rightarrow \lis \; A \rightarrow \lis \; A$ & (\cons -intro) \\ \\
    $\displaystyle {[]} \vdash \listrec : \forall A : \Type_0. \forall F : \lis \; A \rightarrow \Type_0.$ \hspace{29ex} & (\listrec -intro) \\
\hfill$\displaystyle F \; (\nil \; A) \rightarrow (\forall a : A. \forall l : \lis \; A. F \; l \rightarrow F \; (\cons \; A \; a \; l)) \rightarrow \forall l : \lis \; A. F \; l$ \\ \\
    $\displaystyle {[]} \vdash \listind : \forall A : \Type_0. \forall P : \lis \; A \rightarrow \Prop.$ \hspace{30ex} & (\listind -intro) \\
\hfill $\displaystyle P \; (\nil \; A) \rightarrow (\forall a : A. \forall l : \lis \; A. P \; l \rightarrow P \; (\cons \; A \; a \; l)) \rightarrow \forall l : \lis \; A. P \; l$ \\

    \hline \hline
  \end{tabular}
  \label{tab:typing_rule_of_ECC_with_list}
  \end{center}
\end{table}

Now, we define the beta equality for \CCwp.

\begin{definition}[Beta Equality for \CCwp]\label{beta_eccp}
  Let $=_{\beta'}$ be the smallest equivalence relation such that the following conditions hold.
  \begin{enumerate}[(i)]
  \item $(\lambda x : A. t) \; a =_{\beta'} t[x \backslash a]$.
  \item If $t_1 =_{\beta'} t_1'$ and $t_2 =_{\beta'} t_2'$, then $t_1 t_2 =_{\beta'} t_1' t_2'$.
  \item If $t =_{\beta'} t'$ and $A =_{\beta'} A'$, then $\lambda x : A. t =_{\beta'} \lambda x : A' t'$.
  \item If $A =_{\beta'} A'$ and $B =_{\beta'} B'$, then $\forall x : A. B =_{\beta'} \forall x : A' B'$.
  \item $\listrec \; A \; F \; t_1 \; t_2 \; (\nil \; A) =_{\beta'} t_1$
  \item $\listrec \; A \; F \; t_1 \; t_2 \; (\cons \; A \; a \; l) =_{\beta'} t_2 \; a \; l \; (\listrec \; A \; F \; t_1 \; t_2 \; l)$
  \item $\listind \; A \; F \; t_1 \; t_2 \; (\nil \; A) =_{\beta'} t_1$
  \item $\listind \; A \; F \; t_1 \; t_2 \; (\cons \; A \; a \; l) =_{\beta'} t_2 \; a \; l \; (\listind \; A \; F \; t_1 \; t_2 \; l)$
  \end{enumerate}
\end{definition}

Now that we defined \CCwp 's terms and typing rules,
we can define some familiar operators over $\lis$ type, such as membership operator `${\sf in}$'  for instance.
\begin{eqnarray*}
  && {\sf in} : \forall A : \Type_0. A \rightarrow \lis \; A \rightarrow \Prop := \\
  && \quad \lambda A : \Type_0 . \; \lambda a : A. \; \lambda l : \lis A . \\
  && \qquad (\listrec \; A \; (\lambda \_ : \lis \; A. \; \Prop) \\
  && \qquad\quad \mathrm{False} \\
  && \qquad\quad ( \; \lambda x : A. \; \lambda \_ : \lis \; A. \; \lambda {\sf ind} : \Prop. \; x = a \lor {\sf ind} \; ) \\
  && \qquad\quad l )
\end{eqnarray*}
We can then derive the following equalities from definition \ref{beta_eccp}.
\begin{itemize}
\item ${\sf in} \; A \; a \; (\nil \; A) =_{\beta'} \mathrm{False}$
\item ${\sf in} \; A \; a \; (\cons \; A \; x \; l) =_{\beta'} x = a \lor {\sf in} \; A \; a \; l$
\end{itemize}

\subsection{Interpretation}
Here, we define an interpretation of \CCwp.
The interpretation of lists is obtained through an initial algebra
construction.
We fix an arbitrary element denoted by the dot symbol `$\cdot$' to
interpret the unit type.
We can then define the interpretations of $\lis$, $\nil$, $\cons$, $\listrec$ and $\listind$ as follows.

\begin{enumerate}[{(}I{)}]

  \item Interpretation of $\lis$. \\
    First, we define the Kleene closure $S^*$ of a set $S$ as follows.
    \begin{equation*}
      S^* := \bigcup_{n \in \omega} S^n
    \end{equation*}
    where $S^n$ is an $n$-tuple of S, i.e.
    \begin{eqnarray*}
      S^0 &:=& \{ (0, \cdot) \} \\
      S^{n+1} &:=& \{(1, (a, l)) \mid a \in S \; \mathrm{and} \; l \in S^n\}.
    \end{eqnarray*}
    Then $\lis$ is interpreted as a function building the Kleene closure of a set.
    \begin{equation*}
      \jump{\Gamma \vdash \lis}(\gamma) := \{(S, S^*) \mid S \in \U_0 \}
    \end{equation*}
    We can easily check that
    \begin{equation*}
      \jump{\Gamma \vdash \lis}(\gamma) \in \jump{\Gamma \vdash \Type_0 \rightarrow \Type_0}(\gamma)
    \end{equation*}
    holds for any $\gamma \in \jump{\Gamma}$.

  \item Interpretation of $\nil$. \\
    $\nil$ is interpreted by the constant function returning `$(0,\cdot)$'.
    \begin{equation*}
      \jump{\Gamma \vdash \nil}(\gamma) := \{(S, (0, \cdot)) \mid S \in \U_0\},
    \end{equation*}
    We can again easily check that
    \begin{equation*}
      \jump{\Gamma \vdash \nil}(\gamma) \in \jump{\Gamma \vdash \forall A : \Type_0, \lis A}(\gamma)
    \end{equation*}
    holds since $(0,\cdot) \in S^*$ for any set $S$.

  \item Interpretation of $\cons$. \\
    First, we define $\cons_S$ as follows
    \begin{equation*}
      \cons_S := \{(s, (l, (1, s, l))) \mid s \in S \; \mathrm{and} \; l \in S^*\}
    \end{equation*}
    for any set $S$.
    We can easily check that
    \begin{equation*}
      \cons_S \in S \rightarrow S^* \rightarrow S^*
    \end{equation*}
    holds.
    Now, we can define the interpretation of $\cons$ as follows.
    \begin{equation*}
      \jump{\Gamma \vdash \cons}(\gamma) := \{(S, \cons_S) \mid S \in \U_0\}
    \end{equation*}
    We can again easily check that
    \begin{equation*}
      \jump{\Gamma \vdash \cons}(\gamma) \in \jump{\Gamma \vdash \forall A : \Type_0, A \rightarrow \lis \; A \; \rightarrow \lis \; A }(\gamma)
    \end{equation*}
    holds.

  \item Interpretation of $\listrec$. \\
    Given a function $T : S^* \rightarrow \U_0$,
    we define the dependent function $\mathrm{rec}^{(n)}_{t, f} \in \prod_{l \in S^n} T(l)$ by recursion on natural numbers.
    \begin{eqnarray*}
      \mathrm{rec}^{(0)}_{t, f} &:=& \{ ((0, \cdot), t) \} \\
      \mathrm{rec}^{(n+1)}_{t, f} &:=& \{ ((1, (a, l)), f(a)(l)(\mathrm{rec}^{(n)}_{t,f}(l))) \mid a \in S \; \mathrm{and} \; l \in S^n \}
    \end{eqnarray*}
    where $t$ is an element of $T((0,\cdot))$ and $f$ is a dependent function
    \begin{equation*}
      f \in \prod_{a \in S} \prod_{l \in S^*} \Bigl( T(l) \rightarrow T((1,(a, l))) \Bigr).
    \end{equation*}
    Next, we define $\mathrm{rec}_{t, f} \in \prod_{l \in S^*} T(l)$ as follows.
    \begin{equation*}
      \mathrm{rec}_{t,f} := \bigcup_{n \in \omega} \mathrm{rec}^{(n)}_{t,f}
    \end{equation*}

    Finally, we define $\listrec$ as follows.
    \begin{eqnarray*}
      \jump{\Gamma \vdash \listrec}(\gamma) &:=& \{ (S, (T, (t, (f, \mathrm{rec}_{t, f})))) \mid \\
      && \qquad S \in \U_0 \\
      && \qquad T \in S^{*} \rightarrow \U_0 \\
      && \qquad t \in T((0,\cdot)) \\
      && \qquad f \in \prod_{a \in S} \prod_{l \in S^*} \Bigl( T(l) \rightarrow T((1,(a, l))) \Bigr) \}
    \end{eqnarray*}
    We can again easily check that
    \begin{eqnarray*}
      && \jump{\Gamma \vdash \listrec}(\gamma) \in \jump{\Gamma \vdash \\
        && \quad \forall A : \Type_0. \forall F : \lis A \rightarrow \Type_0. \\
        && \qquad F (\nil \; A) \rightarrow (\forall a : A. \forall l : \lis A. F \; l \rightarrow F \; (\cons \; A \; a \; l)) \rightarrow \forall l : \lis \; A. F \; l \\
        && }(\gamma)
    \end{eqnarray*}
    holds.

  \item Interpretation of $\listind$. \\
    The interpretation of $\listind$ is much simple.
    Since $\listind$ is a proof term, its interpretation must be
    \begin{equation*}
      \jump{\Gamma \vdash \listind}(\gamma) := \floor{\gamma}.
    \end{equation*}
    For the soundness theorem, we shall prove that
    \begin{eqnarray*}
      && \jump{\Gamma \vdash \listind}(\gamma) \in \jump{\Gamma \vdash \\
      && \quad \forall A : \Type_0. \forall P : \lis A \rightarrow \Prop. \\
      && \qquad P (\nil \; A) \rightarrow (\forall a : A. \forall l : \lis A. P \; l \rightarrow P \; (\cons \; A \; a \; l)) \rightarrow \forall l : \lis \; A. P \; l \\
      && }(\gamma)
    \end{eqnarray*}
    holds.
    This is a corollary of Lemma \ref{interpretation_of_list_ind}.
\end{enumerate}

It remains to prove the soundness of \CCwp.
\begin{theorem}[soundness for \CCwp]
  \begin{enumerate}
  \renewcommand{\labelenumi}{(\arabic{enumi})}
  \item
    If $t_1 =_{\beta'} t_2$ holds and $\Gamma \vdash t_1 : T$ and $\Gamma \vdash t_2 : T$ are derivable, then $\jump{\Gamma \vdash t_1}(\gamma) = \jump{\Gamma \vdash t_2}(\gamma)$ holds.
  \item
    If $\Gamma \vdash t : T$ is derivable in \CCwp, then $\jump{\Gamma \vdash t}(\gamma) \in \jump{\Gamma \vdash T}(\gamma)$ holds.
  \end{enumerate}
  \end{theorem}
To prove (1), we need a \CCwp version of Lemma \ref{proposition_proof_invariant} and Lemma \ref{substitution_interpretation}.
They can be proved similary as for \CCw.
To prove (2), we need the following lemma that states the soundness of
induction on lists.
\begin{lemma}\label{interpretation_of_list_ind}
  \begin{eqnarray*}
    && \jump{\Gamma \vdash \forall A : \Type_0. \forall P : \lis \; A \rightarrow \Prop. \\
      && \qquad P (\nil \; A) \rightarrow (\forall a : A. \forall l : \lis \; A. P \; l \rightarrow P \; (\cons \; A \; a \; l)) \rightarrow \forall l : \lis \; A. P \; l \\
    && }(\gamma) = X
  \end{eqnarray*}
  where $X$ is the whole topological space $(X, \mathcal{O}(X))$.
\end{lemma}
\begin{proof}
  Let S be a set and $\psi \in S^{*} \to \mathcal{O}(X)$ be a function.
  We define the set of open sets $T^{\psi}_{n}$ as
  \begin{eqnarray*}
    T^{\psi}_{0} &:=& \{ \psi(0, \cdot) \} \\
    T^{\psi}_{n+1} &:=& T^{\psi}_n \cup \{ \psi(1, (a, l))^{\psi(l)} \mid a \in S \; \mathrm{and} \; l \in S^n\}.
  \end{eqnarray*}
  For any $n \in \omega$ and $l \in S^n$, $\bigsqcap T^{\psi}_n \leq \psi(l)$ holds by induction on natural numbers.
  Let $T^{\psi}$ be their union
  \begin{equation*}
    T^{\psi} := \bigcup_{n \in \omega} T^{\psi}_n.
  \end{equation*}

  Since $T^{\psi}_n \subset T^{\psi}$, therefore $\bigsqcap T^{\psi} \leq \bigsqcap T^{\psi}_n$ holds, hence we have $\bigsqcap T^{\psi} \leq \psi(l)$ for any $l \in S^{*}$.
  Therefore, we also have $\bigsqcap T^{\psi} \leq \bigsqcap \{ \psi(l) \mid l \in S^{*}\}$.

  Now, let us calculate the first equation:
  \begin{eqnarray*}
    && \jump{\forall A : \Type_0. \forall P : \lis \; A \rightarrow \Prop. \\
      && \qquad P (\nil \; A) \rightarrow (\forall a : A. \forall l : \lis \; A. P \; l \rightarrow P \; (\cons \; A \; a \; l)) \rightarrow \forall l : \lis \; A. P \; l \\
      && }(\gamma) \\
    && = \bigsqcap_{S \in \U_0} \biggl( \bigsqcap_{\psi \in S^{*} \to \mathcal{O}(X)} \bigl(\bigsqcap_{l \in S^{*}} \psi(l) \bigr)^{(\bigsqcap T^{\psi})} \biggr) \\
    && = \bigsqcap_{S \in \U_0} \biggl( \bigsqcap_{\psi \in S^{*} \to \mathcal{O}(X)} X \biggr) \\
    && = X
  \end{eqnarray*}
  To calculate it, we use (\ref{eq:powerprod}) and (\ref{eq:le_whole}) from lemma \ref{heyting_conditions}.

\end{proof}

\section{Conclusion and Future Work}\label{future}
We could construct an intuitionistic set-theoretical model of \ECC,
which allowed us to prove that PEM and the linearity axiom do not hold in \ECC.
This model is not complete with respect to plain \ECC, since it is
proof-irrelevant.

This model combines an impredicative interpretation of propositional
terms and a predicative interpretation of non-propositional terms
as in~\cite{not_simple}.

Since one of our goals is to provide a model for Coq, we need to extend our model to all of CIC.
This requires working on several extensions:
\begin{itemize}
\item CIC adds subsumption between $\Prop$ and $\Type_i$.
\[
\frac{\Gamma \vdash A : \Prop}{\Gamma \vdash A : \Type_i}
\]
In fact, this rule breaks Lemma~\ref{substitution_in_proof} and \ref{product_type_substitution_weakening}.
As a result, Theorem \ref{soundness_of_beta_equality}, soundness of
beta equality, does not hold, as we show here.

Let $\mathcal{I}$ be $\lambda T : \Type_i. T \to T$.
In a set-theoretical interpretation, $\jump{\mathcal{I}}$ must be a function $A \mapsto \{f \mid f : A \to A\}$.
However, for any propositional term $P$, the term $\mathcal{I} P$ is a
tautology, and its interpretation is $X$, which leads to conflicting
interpretations as $\jump{\mathcal{I}}(\jump{P}) = \jump{P} \to \jump{P} \neq X$.
Using an idea from Aczel~\cite{rel_set_type}, Lee and Werner and their
followers~\cite{ModelEq, ModelEq2} avoided this
problem in an elegant way, by giving a uniform interpretation of
propositional and non-propositional terms.
They define the encoding functions $\app$ and $\lam$ as follows.
\begin{eqnarray*}
  \app (u,x) &:=& \{ z \mid (x, z) \in u\} \\
  \lam (f) &:=& \bigcup_{(x,y) \in f} \{(x,z) \mid z \in y\}
\end{eqnarray*}
These sastisfy the expected property $\app(\lam(f),x) = f(x)$.
Using the classical interpretation $\jump{\Prop} = \{\varnothing,
\{\varnothing\}\}$, the interpretation of the product type $\forall x
: A.B$ becomes $\{\lam(f) \mid f \in \prod_{x \in A} B(x)\} \in
\jump{\Prop}$.
It evaluates to $\{\varnothing\}$ iff $B(x) = \{\varnothing\}$ for all
$x \in A$.

Unfortunately, this solution does not apply to intuitionistic
settings, since $\jump{\Prop}$ should contain more elements,
making such a simple encoding impossible.
We believe that searching for a non uniform encoding is a more resonable
direction.

\item Finally, CIC adds inductive and co-inductive type definitions, and they both can live in the impredicative universe $\Prop$.
  Lee and Werner's model \cite{ModelEq} supports generic inductive definitions
  through their set theoretical interpretation in a
  predicative universe as it was defined by Dybjer~\cite{ind_rec_def},
  using Aczel's $\Phi$-closed set approach~\cite{CZF}.
  However, they do not extend this interpretation to the impredicative case.
  We have not yet investigated how to handle generic inductive
  definitions, co-inductive defintions, and impredicative inductive
  definitions in our model.
\end{itemize}

\bibliographystyle{apalike}
\bibliography{model}

\appendix

\section{Proof of Weakening}
\begin{proof}[Lemma~\ref{interpretation_composition}]\label{interpretation_weakening_appendix}
  We prove by induction hypothesis and by Lemma~\ref{weakening_in_proof}.
  \begin{itemize}
  \item $t = x$ (case of variable) \\
    It is clear since $t$ is not a proof term.

  \item $t = \lambda x : A. t'$ \\
    By Lemma~\ref{proposition_proof_invariant}, $t$ is also not a proof term for $\Gamma_1 ; (x' : A') ; \Gamma_2 ; (x : A)$.
    Therefore
    \begin{eqnarray*}
      && \jump{\Gamma_1 ; \Gamma_2 \vdash \lambda x : A.t'}(\gamma_1, \gamma_2) \\
      &=& \{ (\alpha, \jump{\Gamma_1 ; \Gamma_2 ; (x : A) \vdash t'}(\gamma_1, \gamma_2, \alpha) \; | \\
      && \qquad \alpha \in \jump{\Gamma_1 ; \Gamma_2 \vdash A}(\gamma_1, \gamma_2) \} \\
      &=& \{ (\alpha, \jump{\Gamma_1 ; (x' : A') ; \Gamma_2; (x : A) \vdash t'}(\gamma_1, \alpha', \gamma_2, \alpha) \; | \\
      && \qquad \alpha \in \jump{\Gamma_1 ; (x' : A') ; \Gamma_2 \vdash A}(\gamma_1, \alpha', \gamma_2) \} \\
      &=& \jump{\Gamma_1 ; (x' : A') ; \Gamma_2 \vdash \lambda x : A.t'}(\gamma_1, \alpha', \gamma_2)
    \end{eqnarray*}
    holds.

  \item $t = t_1 \; t_2$ \\
    By Lemma~\ref{proposition_proof_invariant}, $t_1$ is also not a proof term for $\Gamma_1 ; (x' : A') ; \Gamma_2$.
    If $t_2$ is a proof term for $\Gamma_1 ; (x' : A') ; \Gamma_2$, then
    \begin{eqnarray*}
      && \jump{\Gamma_1 ; \Gamma_2 \vdash t_1 \; t_2}(\gamma_1, \gamma_2) \\
      &=& \jump{\Gamma_1 ; \Gamma_2 \vdash t_1}(\gamma_1, \gamma_2)\left(\bot_X\right) \\
      &=& \jump{\Gamma_1 ; (x' : A') ; \Gamma_2 \vdash t_1}(\gamma_1, \alpha', \gamma_2)\left(\bot_X\right) \\
      &=& \jump{\Gamma_1 ; (x' : A') ; \Gamma_2 \vdash t_1 \; t_2}(\gamma_1, \alpha', \gamma_2)
    \end{eqnarray*}
    holds.
    If $t_2$ is not a proof term, then it is proved similarly.

  \item $t = \forall x : A.B$ \\
    By Lemma~\ref{product_type_substitution_weakening}
    \begin{equation*}
      \PD{\Gamma_1 ; \Gamma_2}{x}{A}{B} = \PD{\Gamma_1; (x' : A') ; \Gamma_2}{x}{A}{B}
    \end{equation*}
    holds.
    Hence we can only prove in the case of $\PD{\Gamma_1 : \Gamma_2}{x}{A}{B}$.
    
  \item $t = \Prop$ or $\Type_i$ \\
    Clear.
  \end{itemize}
\end{proof}

\section{Proof of Substitution}

\begin{proof}[Lemma~\ref{substitution_interpretation}]\label{interpretation_substitution_appendix}
    We define the predicates $P(\Delta)$ and $Q(\Delta, t)$ as follows.
  \begin{eqnarray*}
    P(\Delta) &\equiv& \forall \delta, (\gamma, \jump{\Gamma \vdash u}(\gamma), \delta) \in \jump{\Gamma ; (x : u) ; \Delta)} \\
    && \Rightarrow \quad (\gamma, \delta) \in \jump{\Gamma; \Delta[x \backslash u]}, \\
    Q(\Delta, t) &\equiv& \forall \delta, \jump{\Gamma ; (x : U) ; \Delta \vdash t}(\gamma, \jump{\Gamma\vdash u}(\gamma), \delta) \mbox{ is well defined} \\
    && \quad \Rightarrow \quad \biggl(\jump{\Gamma ; \Delta[x \backslash u] \vdash t[x \backslash u]}(\gamma, \delta) \mbox{ is well-defined} \\
    && \quad \quad \mbox{and} \quad \jump{\Gamma ; x : U ; \Delta \vdash t}(\gamma, \jump{\Gamma \vdash u}(\gamma), \delta) \\
    && \quad \quad \quad \quad \quad = \jump{\Gamma ; \Delta[x \backslash u] \vdash t[x \backslash u]}(\gamma, \delta)\biggr).
  \end{eqnarray*}
  We prove this lemma in three steps (i)$P([])$, (ii)$P(\Delta)
  \Rightarrow \forall t,Q(\Delta, t)$, (iii)$(\forall t, Q(\Delta, t))
  \Rightarrow \forall T, P(\Delta; y : T)$.
  \begin{enumerate}[{(}i{)}]
  \item $P([])$ \\
    Clear
  \item $P(\Delta) \Rightarrow \forall t, Q(\Delta, t)$ \\
    If $t$ is a proof term for $\Gamma ; (x : U) ; \Delta$, then $t[x \backslash u]$ is also a proof term for $\Gamma ; \Delta[x \backslash u]$ by Lemma~\ref{substitution_in_proof}.
    Therefore
    \begin{eqnarray*}
      \jump{\Gamma ; (x : U) ; \Delta \vdash t}(\gamma, \jump{\Gamma \vdash u}(\gamma), \delta) &=& \floor{\gamma, \jump{\Gamma \vdash u}(\gamma), \delta} \\
      \jump{\Gamma ; \Delta \vdash t[x \backslash u]}(\gamma, \delta) &=& \floor{\gamma, \delta}
    \end{eqnarray*}
    hold.
    Hence we must prove is
    \begin{equation*}
      \floor{\gamma, \jump{\Gamma \vdash u}(\gamma), \delta} = \floor{\gamma, \delta}.
    \end{equation*}
    If $u$ is not a proof term for $\Gamma$, it is clear.
    If $u$ is a proof term then $\jump{\Gamma \vdash u}(\gamma) = \floor{\gamma}$ holds, therefore it also hold.

    Next, if $t$ is not a proof term for $\Gamma ; (x : U) ; \Delta$, then $t[x \backslash u]$ is also not a proof term for $\Gamma ; \Delta[x \backslash u]$ by Lemma~\ref{substitution_in_proof}.
    We prove by induction on the term $t$.
    \begin{itemize}
    \item $t = \Prop$ or $\Type_i$ \\
      It is clear.
    \item $t = \forall a : A.B$ \\
      We assume that
      \begin{equation*}
        \jump{\Gamma ; (x : U) ; \Delta \vdash \forall a : A.B}(\gamma, \jump{\Gamma \vdash u}(\gamma), \delta)
      \end{equation*}
      is well defined, therefore
      \begin{eqnarray*}
        \jump{\Gamma ; (x : U) ; \Delta \vdash A}(\gamma, \jump{\Gamma \vdash u}(\gamma), \delta), \\
        \jump{\Gamma ; (x : U) ; \Delta ; (a : A) \vdash B}(\gamma, \jump{\Gamma \vdash u}(\gamma), \delta, \alpha)
      \end{eqnarray*}
      are also well defined.
      By induction hypothesis, $Q(\Delta, A)$ and $Q(\Delta ; (a : A), B)$ are assumed.
      By Lemma~\ref{product_type_substitution_weakening}, the value of $\PT{}{}$ is invariant.
      Hence the statement holds in this case.
    \item $t = \lambda a : A. t$ \\
      We assume that
      \begin{equation*}
        \jump{\Gamma ; (x : U) ; \Delta \vdash \lambda a : A. t}(\gamma, \jump{\Gamma \vdash u}(\gamma), \delta)
      \end{equation*}
      is well defined, therefore
      \begin{eqnarray*}
        \jump{\Gamma ; (x : U) ; \Delta ; (a : A) \vdash t}(\gamma, \jump{\Gamma \vdash t}(\gamma), \delta, \alpha) \\
        \jump{\Gamma ; (x : U) ; \Delta \vdash A}(\gamma, \jump{\Gamma \vdash u}(\gamma), \delta)
      \end{eqnarray*}
      are also well defined.
      By induction hypothesis, $Q(\Delta ; (a : A), t)$ and $Q(\Delta, A)$ are assumed.
      Hence the statement holds in this case.      
    \item $t = a \; b$ \\
      We assume that
      \begin{equation*}
        \jump{\Gamma ; (x : U) ; \Delta \vdash a \; b}(\gamma, \jump{\Gamma \vdash u}(\gamma), \delta)
      \end{equation*}
      is well defined, therefore
      \begin{equation*}
        \jump{\Gamma ; (x : U) ; \Delta \vdash a}(\gamma, \jump{\Gamma \vdash u}(\gamma), \delta)
      \end{equation*}
      is well defined and a function whose domain contains
      \begin{equation*}
        \jump{\Gamma ; (x : U) ; \Delta \vdash b}(\gamma, \jump{\Gamma \vdash u}(\gamma), \delta).
      \end{equation*}
      By induction hypothesis, $Q(\Delta, a)$ and $Q(\Delta, b)$ are assumed.
      By Lemma~\ref{substitution_in_proof}, if $b$ is a (resp. not) proof term for $\Gamma ; (x : U) ; \Delta$, then $b[x \backslash u]$ is also a (resp. not) proof term for $\Gamma ; \Delta[x \backslash u]$. 
      Hence the statement holds in this case.
    \item $t = y$ (case of variable) \\
      We prove in three cases as follows.
      \begin{itemize}
      \item The variable $y$ occur in $\Gamma$. \\
        In this case, we have
        \begin{eqnarray*}
          \jump{\Gamma ; (x : U) ; \Delta \vdash y}(\gamma, \jump{\Gamma \vdash u}(\gamma), \delta) &=& \gamma_i, \\
          \jump{\Gamma ; \Delta[x \backslash u] \vdash y[x \backslash u]}(\gamma, \delta) &=& \gamma_i.
        \end{eqnarray*}
        for some $i$.
        Hence the statement holds in this case.

      \item The case $y = x$. \\
        We have
        \begin{eqnarray*}
          \jump{\Gamma ; (x : U) ; \Delta \vdash x}(\gamma, \jump{\Gamma \vdash u}(\gamma), \delta) = \jump{\Gamma \vdash u}(\gamma), \\
          \jump{\Gamma ; \Delta[x \backslash u] \vdash x[x \backslash u]}(\gamma, \delta) = \jump{\Gamma ; \Delta[x \backslash u] \vdash u}(\gamma).
        \end{eqnarray*}
        By Lemma~\ref{interpretation_composition}, the statement holds in this case.
        
      \item The variable $y$ occur in $\Delta$. \\
        In this case, we have
        \begin{eqnarray*}
          \jump{\Gamma ; (x : U) ; \Delta \vdash y}(\gamma, \jump{\Gamma \vdash u}(\gamma), \delta) &=& \delta_i
        \end{eqnarray*}
        for some $i$.
        Since $(\gamma, \delta) \in \jump{\Gamma ; \Delta[x \backslash u]}$ by hypothesis $P(\Delta)$, hence following equation is well defined.
        \begin{equation*}
          \jump{\Gamma ; \Delta[x \backslash u] \vdash y}(\gamma, \delta) = \delta_i
        \end{equation*}
        Hence the statement holds in this case.
      \end{itemize}

    \end{itemize}
  \item $(\forall t, Q(\Delta, t)) \Rightarrow \forall T, P(\Delta; y : T)$ \\
    We assume that
    \begin{equation*}
      (\gamma, \jump{\Gamma \vdash u}(\gamma), \delta, \epsilon) \in \jump{\Gamma ; (x : U) ; \Delta ; (y : T)}.
    \end{equation*}
    By definition of interpretaion of context, we have
    \begin{eqnarray*}
      (\gamma, \jump{\Gamma \vdash u}(\gamma), \delta) &\in& \jump{\Gamma ; (x : U) ; \Delta} \\
      \epsilon &\in& \jump{\Gamma ; (x : U) ; \Delta \vdash T}(\gamma, \jump{\Gamma \vdash u}(\gamma), \delta)
    \end{eqnarray*}
    Since $Q(\Delta, T)$ holds, hence following equations hold.
    \begin{eqnarray*}
      (\gamma, \delta) &\in& \jump{\Gamma ; \Delta[x \backslash u]}, \\
      \epsilon &\in& \jump{\Gamma ; \Delta[x \backslash u] \vdash T[x \backslash u]}(\gamma, \delta).
    \end{eqnarray*}
    Therefore we have
    \begin{equation*}
      (\gamma, \delta, \epsilon) \in \jump{\Gamma ; \Delta[x \backslash u] \vdash T[x \backslash u]}(\gamma, \delta).
    \end{equation*}
  \end{enumerate}
\end{proof}

\section{Proof of semantic proof irrelevance}

\begin{proof}[Lemma~\ref{interpretation_constant}]\label{interpretation_constant_appendix}
  \
  \begin{itemize}
  \item $t = x$ (case of variable) \\
    Since $t$ is not a proof term for $\Gamma ; (x' : A') ; \Delta$, therefore we have $t \neq x'$, hence the statement holds in this case.

  \item $t = \lambda x : A. t'$ \\
    By Lemma~\ref{proposition_proof_invariant}, $t$ is also not a proof term for $\Gamma ; (x' : A') ; \Delta ; (x : A)$.
    Therefore
    \begin{eqnarray*}
      && \jump{\Gamma ; (x' : A') ; \Delta \vdash \lambda x : A.t'}(\gamma, p_1, \delta) \\
      &=& \{ (\alpha , \jump{\Gamma ; (x' : A') ; \Delta ; (x : A) \vdash t'}(\gamma, p_1, \delta, \alpha)) \; | \\
      && \qquad \; \alpha \in \jump{\Gamma ; (x' : A') ; \Delta \vdash A}(\gamma, p_1, \delta) \} \\
      &=& \{ (\alpha , \jump{\Gamma ; (x' : A') ; \Delta ; (x : A) \vdash t'}(\gamma, p_2, \delta, \alpha)) \; | \\
      && \qquad \; \alpha \in \jump{\Gamma ; (x' : A') ; \Delta \vdash A}(\gamma, p_2, \delta) \} \\
      &=& \jump{\Gamma ; (x' : A') ; \Delta \vdash \lambda x : A.t'}(\gamma, p_2, \delta)
    \end{eqnarray*}
    holds.
    
  \item $t = t_1 \; t_2$ \\
    By Lemma~\ref{proposition_proof_invariant}, $t_1$ is also not a proof term for $\Gamma ; (x' : A') ; \Delta$.
    If $t_2$ is a proof term for $\Gamma ; (x' : A') ; \Delta$, then
    \begin{eqnarray*}
      && \jump{\Gamma ; (x' : A') ; \Delta \vdash t_1 \; t_2}(\gamma, p_1, \delta) \\
      &=& \jump{\Gamma ; (x' : A') ; \Delta \vdash t_1}(\gamma, p_1, \delta)(\bot_X)\\
      &=& \jump{\Gamma ; (x' : A') ; \Delta \vdash t_1}(\gamma, p_2, \delta)(\bot_X) \\
      &=& \jump{\Gamma ; (x' : A') ; \Delta \vdash t_1 \; t_2}(\gamma, p_2, \delta)
    \end{eqnarray*}
    holds.
    If $t_2$ is not a proof term for $\Gamma ; (x' : A') ; \Delta$, then similarly.

  \item $t = \forall x : A.B$ \\
    Similarly.
  \end{itemize}
\end{proof}

\section{Proof of Soundness}

\begin{proof}[Theorem \ref{soundness}]\label{soundness_ind}
  \
  \begin{enumerate}
  \item \ Case of Axiom \\
    $\jump{\Prop} \in \jump{\Type_i}$ is holds by the condition of $(X, \mathcal{O}(X))$.

  \item \ Case of Weakening \\
    It holds by Lemma~\ref{interpretation_composition}.
    


  \item \ Case of Subsumption \\
    It holds by (\ref{enum:univ_subset}) of Lemma~\ref{univ_close}.
  \item \ Case of PI-Type \\
    We will show the fact that
    \begin{eqnarray*}
      &\bigl(\forall \gamma,\alpha,& \jump{\Gamma \vdash A}(\gamma) \in \jump{\Gamma \vdash s_1}(\gamma) \\
      &\land& \jump{\Gamma ; (x : A) \vdash B}(\gamma, \alpha) \in \jump{\Gamma ; (x : A) \vdash s_2}(\gamma, \alpha)\bigr) \\
      &\Rightarrow& (\forall \gamma, \jump{\Gamma \vdash \forall x:A.B}(\gamma) \in \jump{\Gamma \vdash s_3}(\gamma)).
    \end{eqnarray*}
    There are four cases as follows.
    \begin{itemize}
    \item $\PD{\Gamma}{x}{A}{B} = \TT$\\
      By definition of the interpretation of judgment, the following equation
      \begin{equation*}
        \jump{\Gamma \vdash \forall x : A.B}(\gamma) = \prod_{\alpha \in \jump{\Gamma \vdash A}(\gamma)}\jump{\Gamma ; (x : A) \vdash B}(\gamma, \alpha)
      \end{equation*}
      holds.
      Since $\jump{\Gamma \vdash A}(\gamma) \in \U_i$ , $\jump{\Gamma ; (x : A) \vdash B}(\gamma, \alpha) \in \U_i$ for any $\gamma, \alpha$ and Lemma \ref{univ_close}~(\ref{enum:univ_prod}), we have
        \begin{equation*}
          \displaystyle\prod_{\alpha \in \jump{\Gamma \vdash A}(\gamma)} \jump{\Gamma ; (x : A) \vdash B}(\gamma, \alpha) \in \U_i.
        \end{equation*}

      \item $\PD{\Gamma}{x}{A}{B} = \PT$\\
      By definition of the interpretation of judgment, the following equation
      \begin{eqnarray*}
        && \jump{\Gamma \vdash \forall x : A.B}(\gamma) =\\
        &&\qquad \biggl\{ f \in \prod_{\alpha \in \jump{\Gamma \vdash A}(\gamma)}\jump{\Gamma ; (x : A) \vdash B}(\gamma, \alpha) \; | \\
        && \qquad \qquad \; f \mbox{ is a constant function}\biggr\}
      \end{eqnarray*}
      holds.
      Since $\jump{\Gamma \vdash A}(\gamma) \in \U_i$ , $\jump{\Gamma ; (x : A) \vdash B}(\gamma, \alpha) \in \U_i$ for any $\gamma, \alpha$ and Lemma \ref{univ_close}~(\ref{enum:univ_prod}),
      the statement holds.

    \item $\PD{\Gamma}{x}{A}{B} = \TP$ \\
      It is clear since $\jump{\Gamma \vdash \forall x : A.B}(\gamma)$ is an open set by definition of the interpretation of judgment.
    \item $\PD{\Gamma}{x}{A}{B} = \PP$ \\
      It is clear since $\jump{\Gamma \vdash \forall x : A.B}(\gamma)$ is an open set by definition of the interpretation of judgment.

    \end{itemize}

  \item \ Case of Abstract \\
    We will show the fact that
    \begin{eqnarray*}
      &\bigl(\forall \gamma, \alpha,& \jump{\Gamma ; (x : A) \vdash t}(\gamma, \alpha) \in \jump{\Gamma ; (x : A) \vdash B}(\gamma, \alpha)\bigr) \\
      &\Rightarrow& \bigl(\forall \gamma, \jump{\Gamma \vdash \lambda x : A.t}(\gamma) \in \jump{\Gamma \vdash \forall x : A.B}(\gamma)\bigr).
    \end{eqnarray*}
    There are four cases as follows.
    \begin{itemize}
    \item $\PD{\Gamma}{x}{A}{B} = \TT$ \\
      By definition of the interpretation, we have the following equations:
      \[ \begin{array}{l}
        \displaystyle\jump{\Gamma \vdash \lambda x : A.t}(\gamma) = \\
        \displaystyle\quad \Bigl\{\bigl(\alpha, \jump{\Gamma ; (x : A) \vdash t}(\gamma, \alpha) \bigr) \; | \; \alpha \in \jump{\Gamma \vdash A}(\gamma) \Bigr\}, \\
        \displaystyle\jump{\Gamma \vdash \forall x : A.B}(\gamma) = \\
        \displaystyle\quad\prod_{\alpha \in \jump{\Gamma \vdash A}(\gamma)}\jump{\Gamma ; (x : A) \vdash B}(\gamma, \alpha).
      \end{array} \]
      Then, we must prove the following equation:
      \[ \begin{array}l
        \displaystyle \Bigl\{\bigl(\alpha, \jump{\Gamma ; (x : A) \vdash t}(\gamma,
           \alpha) \bigr) \; | \; \alpha \in \jump{\Gamma \vdash
           A}(\gamma) \Bigr\} \\
           \displaystyle \hfill \in \prod_{\alpha \in \jump{\Gamma \vdash A}(\gamma)}\jump{\Gamma ; (x : A) \vdash B}(\gamma, \alpha).
      \end{array} \]
      But it is clear\footnote{
        If $\jump{\Gamma \vdash A}(\gamma)$ is the empty set, then $\jump{\Gamma \vdash \forall x :A.B}(\gamma) = \{ \varnothing \}$ and $\jump{\Gamma \vdash \lambda x : A.t}(\gamma) = \varnothing$ hold.} by induction of hypothesis.

    \item $\PD{\Gamma}{x}{A}{B} = \PT$ \\
      It is similar the case of $\PD{\Gamma}{x}{A}{B} = \TT$.
      We must prove that
      \[ \begin{array}{l}
        \displaystyle\jump{\Gamma \vdash \lambda x : A.t}(\gamma) = \\
        \displaystyle\quad \Bigl\{\bigl(\alpha, \jump{\Gamma ; (x : A) \vdash t}(\gamma, \alpha) \bigr) \; | \; \alpha \in \jump{\Gamma \vdash A}(\gamma) \Bigr\}      \end{array} \]
      is a constant function.
      Since $A$ is a propositional term for $\Gamma$, this is a
      consequence of Lemma~\ref{interpretation_constant}.

    \item $\PD{\Gamma}{x}{A}{B} = \TP$ \\
      Since $\lambda x : A.t$ is a proof term, we have following equations
      \begin{equation*}
        \jump{\Gamma \vdash \lambda x : A.t}(\gamma) = \floor{\gamma}.
      \end{equation*}
      Hence, the fact we must prove is 
      \begin{equation*}
        \floor{\gamma} \in \jump{\Gamma \vdash \forall x : A.B}(\gamma).
      \end{equation*}
      By definition we have the following equation.
      \[\begin{array}l
        \displaystyle \jump{\Gamma \vdash \forall x : A.B}(\gamma) = \\
        \displaystyle \quad \bigsqcap\{ \jump{\Gamma ; (x : A) \vdash B}(\gamma, \alpha) \; | \; \alpha \in \jump{\Gamma \vdash A}(\gamma) \}.
      \end{array}\]
      If $\jump{\Gamma \vdash A}(\gamma)$ is the empty set, then the statement holds since $\jump{\Gamma \vdash \forall x : A.B}(\gamma) = X$ by Lemma~\ref{heyting_conditions}~(\ref{eq:whole}).
      We assume that $\jump{\Gamma \vdash A}(\gamma)$ is a non-empty set.
      We have
      \begin{equation*}
        \forall \alpha \in \jump{\Gamma \vdash A}(\gamma), \; \floor{\gamma} \in \jump{\Gamma ; (x : A) \vdash B}(\gamma, \alpha).
      \end{equation*}
      since $\jump{\Gamma ; (x : A) \vdash t}(\gamma, \alpha) = \floor{\gamma, \alpha} = \floor{\gamma}$.
      Therefore, we have the following equation
      \begin{equation*}
        \floor{\gamma} \in \bigsqcap\{ \jump{\Gamma ; (x : A) \vdash B}(\gamma, \alpha) \; | \; \alpha \in \jump{\Gamma \vdash A}(\gamma) \}.
      \end{equation*}
      Hence the statement holds in this case.

    \item $\PD{\Gamma}{x}{A}{B} = \PP$ \\
      Since $\lambda x : A.B$ is a proof term, we have the following equation
      \begin{equation*}
        \jump{\Gamma \vdash \lambda x : A.t}(\gamma) = \floor{\gamma}.
      \end{equation*}
      Hence, the fact we must prove is 
      \begin{equation*}
        \floor{\gamma} \in \jump{\Gamma \vdash \forall x : A.B}(\gamma).
      \end{equation*}
      To prove it, we show that
      \begin{equation*}
        \downarrow \floor{\gamma} \; \subset \; \jump{\Gamma \vdash \forall x : A.B}(\gamma).
      \end{equation*}
      This fact is equivalent to the following equation
      \begin{equation*}
        \downarrow \floor{\gamma} \cap \jump{\Gamma \vdash A}(\gamma) \; \subset \; \bigsqcap_{\alpha \in \jump{\Gamma \vdash A}(\gamma)} \jump{\Gamma ; (x : A) \vdash B}(\gamma, \alpha)
      \end{equation*}
      since definition of interpretation and Heyting Algebra.
      We assume $\varepsilon \in \downarrow \floor{\gamma} \cap \jump{\Gamma \vdash A}(\gamma)$.
      By Lemma~\ref{interpretation_constant}, we have
      \begin{equation*}
        \bigsqcap_{\alpha \in \jump{\Gamma \vdash A}(\gamma)} \jump{\Gamma ; (x : A) \vdash B}(\gamma, \alpha) = \jump{\Gamma ; (x : A) \vdash B}(\gamma, \varepsilon)
      \end{equation*}
      holds; since $\varepsilon \in \jump{\Gamma \vdash A}(\gamma)$ holds, right side of this equation is well defined.
      Here, we also have
      \begin{equation*}
        \floor{\gamma, \varepsilon} \in \jump{\Gamma ; (x : A) \vdash B}(\gamma, \varepsilon)
      \end{equation*}
      by induction hypothesis.
      Now, we prove that $\floor{\gamma, \varepsilon} = \varepsilon$ holds.
      Since $\varepsilon \in \downarrow \floor{\gamma}$ holds, therefore we have $\varepsilon \leq \floor{\gamma}$.
      Hence we have $\varepsilon = \floor{\gamma, \varepsilon}$, and the statement holds in this case.
      
    \end{itemize}

  \item \ Case of Apply \\
    We will show the fact that
    \begin{eqnarray*}
      &\bigl(\forall \gamma,& \jump{\Gamma \vdash u}(\gamma) \in
         \jump{\Gamma \vdash \forall x : A.B}(\gamma) \\
      & \land & \jump{\Gamma \vdash v}(\gamma) \in \jump{\Gamma \vdash A}(\gamma)\bigr) \\
      & \Rightarrow & \bigl(\forall \gamma, \jump{\Gamma \vdash u \; v}(\gamma) \in \jump{\Gamma \vdash B[x \backslash v]}(\gamma)\bigr).
    \end{eqnarray*}
    There are four cases as follows.
    \begin{itemize}
    \item $\PD{\Gamma}{x}{A}{B} = \TT$ \\
      By definition of the interpretation of judgment and induction hypothesis, the following equation
      \begin{eqnarray*}
        \jump{\Gamma \vdash u \; v}(\gamma) &=& \jump{\Gamma \vdash u}(\gamma) \bigl(\jump{\Gamma \vdash v}(\gamma) \bigr) \\
        \jump{\Gamma \vdash u}(\gamma) &\in& \prod_{\alpha \in \jump{\Gamma \vdash A}(\gamma)} \jump{\Gamma ; (x : A) \vdash B}(\gamma, \alpha) \\
        \jump{\Gamma \vdash v}(\gamma) &\in& \jump{\Gamma \vdash A}(\gamma)
      \end{eqnarray*}
      hold.
      Therefore, we have
      \begin{equation*}
        \jump{\Gamma \vdash u \; v}(\gamma) \in \jump{\Gamma ; (x : A) \vdash B}(\gamma, \jump{\Gamma \vdash v}(\gamma)).
      \end{equation*}
      By Lemma \ref{substitution_interpretation}, we have
      \begin{equation*}
        \jump{\Gamma ; (x : A) \vdash B}(\gamma, \jump{\Gamma \vdash v}(\gamma)) = \jump{\Gamma \vdash B[x \backslash v]}(\gamma).
      \end{equation*}
      Hence, the statement holds in this case.

    \item $\PD{\Gamma}{x}{A}{B} = \PT$ \\
      By definition of the interpretation of judgment and indcution hypothesis, the following equation
      \begin{eqnarray*}
        \jump{\Gamma \vdash u \; v}(\gamma) &=& \jump{\Gamma \vdash u}(\gamma) \bigl(\bot_X \bigr) \\
        \jump{\Gamma \vdash u}(\gamma) &\in& \biggl\{f \in \prod_{\alpha \in \jump{\Gamma \vdash A}(\gamma)} \jump{\Gamma ; (x : A) \vdash B}(\gamma, \alpha) \; | \\
        && \qquad f \mbox{ is a constant function} \biggr\} \\
        \jump{\Gamma \vdash v}(\gamma) &\in& \jump{\Gamma \vdash A}(\gamma)
      \end{eqnarray*}
      hold.
      Therefore, we have
      \begin{eqnarray*}
        \jump{\Gamma \vdash u \; v}(\gamma) &\in& \jump{\Gamma ; (x : A) \vdash B}(\gamma, \bot_X) \\
        &=& \jump{\Gamma ; (x : A) \vdash B}(\gamma, \jump{\Gamma \vdash v}(\gamma))
      \end{eqnarray*}
      by Lemma~\ref{interpretation_constant}.
      Moreover, the following equation
      \begin{equation*}
        \jump{\Gamma ; (x : A) \vdash B}(\gamma, \jump{\Gamma \vdash v}(\gamma)) = \jump{\Gamma \vdash B[x \backslash v]}(\gamma).
      \end{equation*}
      holds by Lemma \ref{substitution_interpretation}
      Hence, the statement holds in this case.

    \item $\PD{\Gamma}{x}{A}{B} = \TP$ \\
      It suffices to show that $\floor{\gamma} \in \jump{\Gamma \vdash B[x \backslash v]}(\gamma)$, since $\jump{\Gamma \vdash u}(\gamma) = \jump{\Gamma \vdash u \; v}(\gamma) = \floor{\gamma}$ holds.
      By induction hypothesis, we have the following equation
      \begin{equation*}
        \floor{\gamma} \in \bigsqcap\{ \jump{\Gamma ; (x : A) \vdash B}(\gamma, \alpha) \; | \; \alpha \in \jump{\Gamma \vdash A}(\gamma) \}.
      \end{equation*}
      This equation implies the fact that
      \begin{equation*}
        \forall \alpha \in \jump{\Gamma \vdash A}(\gamma), \; \floor{\gamma} \in \jump{\Gamma ; (x : A) \vdash B}(\gamma, \alpha).
      \end{equation*}
      By Lemma \ref{substitution_interpretation} and the fact $\jump{\Gamma \vdash v}(\gamma) \in \jump{\Gamma \vdash A}(\gamma)$, we have
      \begin{equation*}
        \floor{\gamma} \in \jump{\Gamma \vdash B [x \backslash v]}(\gamma).
      \end{equation*}
      Hence, the statement holds in this case.

    \item $\PD{\Gamma}{x}{A}{B} = \PP$ \\
      By induction hypothesis, we have
      \begin{eqnarray*}
        \floor{\gamma} &\in& \jump{\Gamma \vdash \forall x : A.B}(\gamma), \\
        \floor{\gamma} &\in& \jump{\Gamma \vdash A}(\gamma)
      \end{eqnarray*}
      since $\jump{\Gamma \vdash u}(\gamma) = \jump{\Gamma \vdash v}(\gamma)$ holds.
      The following equation holds.
      \begin{equation*}
        \jump{\Gamma \vdash \forall x : A.B}(\gamma) = \biggl(\bigsqcap_{\alpha \in \jump{\Gamma \vdash A}(\gamma)} \jump{\Gamma ; (x : A) \vdash B}(\gamma, \alpha) \biggr)^{\jump{\Gamma \vdash A}(\gamma)}    
      \end{equation*}
      By (\ref{eq:impl}) in Lemma~\ref{heyting_conditions}, we have
      \begin{eqnarray*}
        && \jump{\Gamma \vdash \forall x : A.B}(\gamma) \; \cap \; \jump{\Gamma \vdash A}(\gamma) \\
        && \quad \subset \; \bigsqcap_{\alpha \in \jump{\Gamma \vdash A}(\gamma)}\jump{\Gamma ; (x : A) \vdash B}(\gamma, \alpha).
      \end{eqnarray*}
      Then we also have
      \begin{equation*}
        \floor{\gamma} \in \bigsqcap_{\alpha \in \jump{\Gamma \vdash A}(\gamma)}\jump{\Gamma ; (x : A) \vdash B}(\gamma, \alpha).
      \end{equation*}
      Hence
      \begin{equation*}
        \floor{\gamma} \in \jump{\Gamma ; (x : A) \vdash B}(\gamma, \jump{\Gamma \vdash v}(\gamma))
      \end{equation*}
      holds.
      By Lemma~\ref{substitution_interpretation} and $\jump{\Gamma \vdash u \; v}(\gamma) = \floor{\gamma}$, the statement holds in this case.
    \end{itemize}

  \item \ Case of Variable \\
    We show that
    \begin{equation*}
      \forall \alpha \in \jump{\Gamma \vdash A}(\gamma), \jump{\Gamma ; (x : A) \vdash x}(\gamma, \alpha) \in \jump{\Gamma ; (x : A) \vdash A}(\gamma, \alpha).
    \end{equation*}
    By Lemma~\ref{interpretation_composition}, we must prove is
    \begin{equation*}
      \forall \alpha \in \jump{\Gamma \vdash A}(\gamma), \jump{\Gamma ; (x : A) \vdash x}(\gamma, \alpha) \in \jump{\Gamma \vdash A}(\gamma).
    \end{equation*}
    
    If $A$ is not a propositional term for $\Gamma$, the statement holds since $\jump{\Gamma ; (x : A) \vdash x}(\gamma, \alpha) = \alpha$.
    If $A$ is a propositional term for $\Gamma$, then
    \begin{equation*}
      \jump{\Gamma ; (x : A) \vdash x}(\gamma, \alpha) = \floor{\gamma, \alpha}
    \end{equation*}
    holds.
    Since $\floor{\gamma, \alpha} \; \in \; \downarrow \alpha \subset \jump{\Gamma \vdash A}(\gamma)$,
    \begin{equation*}
      \jump{\Gamma ; (x : A) \vdash x}(\gamma, \alpha) \in \jump{\Gamma \vdash A}(\gamma)
    \end{equation*}
    holds.
    Hence the statement holds in this case.

  \item \ Case of Beta Equality \\
    We must show that
    \begin{eqnarray*}
      &\forall \gamma,& \jump{\Gamma \vdash t}(\gamma) \in
         \jump{\Gamma \vdash A}(\gamma), \; \jump{\Gamma \vdash
         B}(\gamma) \in \jump{\Gamma \vdash s}(\gamma) \\
      & \land & A =_\beta B \\
      & \Rightarrow & \forall \gamma, \jump{\Gamma \vdash t}(\gamma) \in \jump{\Gamma \vdash B}(\gamma).
    \end{eqnarray*}
    It is clear by Theorem~\ref{soundness} (1).
  \end{enumerate}
\end{proof}

\end{document}